\documentclass{lmcs}
\pdfoutput=1
\usepackage[utf8]{inputenc}

\usepackage{lastpage}
\lmcsdoi{19}{1}{12}
\lmcsheading{}{\pageref{LastPage}}{}{}%
{Nov.~02,~2021}{Feb.~14,~2023}{}


\usepackage{hyperref}
\usepackage{framed}
\usepackage{import}
\usepackage{proof}
\usepackage{latexsym}

\newcommand\hide[1]{}

\newcommand\lpr{$\lambda\Pi\slash{\equiv}$}

\def\TYPE{\mbox{\tt TYPE}}
\def\KIND{\mbox{\tt KIND}}
\def\ra{\rightarrow}
\def\lra{\hookrightarrow}
\def\I{{\color{blue} I}}
\def\Set{{\color{blue} \mbox{\it Set}}}
\def\bliota{{\color{blue} \iota}}
\def\El{{\color{blue} \mbox{\it El}}}
\def\Prop{{\color{blue} \mbox{\it Prop}}}
\def\Prf{{\color{blue} \mbox{\it Prf}}}
\def\imp{\mathbin{\color{blue} \Rightarrow}}
\def\fa{{\color{blue} \forall}}
\def\bltop{{\color{blue} \top}}
\def\blbot{{\color{blue} \bot}}
\DeclareMathOperator{\blneg}{{\color{blue} \neg}}
\def\conj{\mathbin{\color{blue} \wedge}}
\def\disj{\mathbin{\color{blue} \vee}}
\def\ex{{\color{blue} \exists}}
\def\S{{\color{blue} \mbox{\it succ}}}
\def\0{{\color{blue} 0}}
\def\positive{{\color{blue} \mbox{\it positive}}}
\def\pred{{\color{blue} \mbox{\it pred}}}
\def\Prfc{{\color{blue} \mbox{\it Prf}_c}}
\def\impc{\mathbin{\color{blue} \Rightarrow_c}}
\def\conjc{\mathbin{\color{blue} \wedge_c}}
\def\disjc{\mathbin{\color{blue} \vee_c}}
\def\fac{{\color{blue} \forall_c}}
\def\exc{{\color{blue} \exists_c}}
\def\o{{\color{blue} o}}
\def\arr{\mathbin{\color{blue} \leadsto}}
\def\arrd{\mathbin{\color{blue} \leadsto_d}}
\def\impd{\mathbin{\color{blue} \Rightarrow_d}}
\def\blpi{{\color{blue} \pi}}
\def\Scheme{{\color{blue} \mbox{\it Scheme}}}
\def\Els{{\color{blue} \mbox{\it Els}}}
\def\lift{{\color{blue} \uparrow}}
\def\calfa{{\rotatebox[origin=c]{180}{{\color{blue} A}}}}
\def\sffa{{\rotatebox[origin=c]{180}{{\color{blue}\ensuremath{\mathcal{A}}}}}}
\def\pos{\mbox{\it positive}}
\def\psub{{\color{blue} \mbox{\it psub}}}
\def\pair{{\color{blue} \mbox{\it pair}}}
\def\paird{{\color{blue} \mbox{\it pair}^\dagger}}
\def\dtarr{\mathbin{\color{blue} \rightsquigarrow_d}}
\def\set{{\color{blue} \mbox{\it set}}}
\def\Sett{{\color{blue} \mbox{\it Set1}}}
\def\Ty{{\color{blue} \mbox{\it Ty}}}

\def\fst{{\color{blue} \mbox{\it fst}}}
\def\snd{{\color{blue} \mbox{\it snd}}}

\def\const{{\mbox{\it const}}}

\newcounter{WarnCounts}

\newcounter{ToDos}

\newcommand\cC{{\mathcal C}}

\newcommand\cR{{\mathcal R}}
\newcommand\cE{{\mathcal E}}
\newcommand\cS{{\mathcal S}}
\newcommand\cM{{\mathcal M}}

\newcommand\cU{{\mathcal U}}
\newcommand\cV{{\mathcal V}}
\newcommand\cL{{\mathcal L}}

\newcommand{\dedukti}{\textsc{Dedukti}}

\newcommand{\lpc}{{\ensuremath{\lambda \Pi\text{-calculus}}}}
\newcommand{\lpcm}{{\ensuremath{\lambda \Pi\text{-calculus modulo theory}}}}
\newcommand{\lprolog}{{\ensuremath{\lambda\text{-Prolog}}}}

\newcommand{\tabs}[3]{\ensuremath{\lambda{#1}\,{:}\,{#2},{#3}}}

\newcommand{\tapp}[2]{\ensuremath{{#1}~{#2}}}

\newcommand{\tpi}[3]{\ensuremath{\Pi{#1}\,{:}\,{#2},#3}}

\newcommand\irule[3]{\infer[#3]{#2}{#1}}

\title{A modular construction of type theories}

\author[F.~Blanqui]{Frédéric Blanqui\lmcsorcid{0000-0001-7438-5554}}[a]
\author[G.~Dowek]{Gilles Dowek\lmcsorcid{0000-0001-6253-935X}}[a]
\author[E.~Grienenberger]{Emilie Grienenberger}[a]
\author[G.~Hondet]{Gabriel Hondet}[a]
\author[F.~Thiré]{François Thiré}[b]

\address{Universit\'e Paris-Saclay, ENS Paris-Saclay, LMF, CNRS, Inria, France}
\email{\{frederic.blanqui,gilles.dowek,gabriel.hondet\}@inria.fr}
\email{emilie.grienenberger@ens-paris-saclay.fr}
\address{Nomadic Labs, France}
\email{francois.thire@nomadic-labs.com}

\begin{document}

\maketitle

\begin{abstract}
  The $\lambda\Pi$-calculus modulo theory is a logical framework in
  which many type systems can be expressed as theories.  We present
  such a theory, the theory $\cU$, where proofs of several logical
  systems can be expressed. Moreover, we identify a sub-theory of
  $\cU$ corresponding to each of these systems, and prove that, when a
  proof in $\cU$ uses only symbols of a sub-theory, then it is a proof
  in that sub-theory.
\end{abstract}

\section{Introduction}

The \lpcm{} (\lpr)~\cite{cousineau07tlca}, that is the basis of the language
\dedukti{}
\cite{expressing,hondet20fscd}, is a logical framework, that is, a framework to
define theories. It generalizes some previously proposed frameworks:
  Predicate logic \cite{HilbertAckermann}, \lprolog{}
\cite{NadathurMiller}, Isabelle \cite{Paulson}, the
  Edinburgh logical framework \cite{HHP}, also called the \lpc,
  Deduction modulo theory \cite{dowek03jar,dowek03jsl}, Pure type systems
\cite{Berardi,Terlouw}, and Ecumenical logic
\cite{Prawitz15,Dowek15,PereiraRodriguez17,Grienenberger19}.  It is
thus an extension of Predicate logic that provides the
possibility for all symbols to bind variables, a syntax for
proof-terms, a notion of computation, a notion of proof reduction for
axiomatic theories, and the possibility to express both constructive
and classical proofs.

\lpr{} enables to express all theories that can be expressed in
  Predicate logic, such as geometry, arithmetic, and set theory, but
also Simple type theory \cite{Church40} and the Calculus of
constructions \cite{CoquandHuet88}, that are less easy to define in
Predicate logic.

We present a theory in \lpr{}, the theory $\cU$, where all proofs of
Minimal, Constructive, and Ecumenical predicate logic; Minimal,
Constructive, and Ecumenical simple type theory; Simple type theory
with predicate subtyping, prenex predicative polymorphism, or both;
the Calculus of constructions, and the Calculus of constructions with
prenex predicative polymorphism can be expressed.  This theory is
therefore a candidate for a universal theory, where proofs developed
in implementations of Classical predicate logic (such as automated
theorem proving systems, SMT solvers, etc.), Classical simple type
theory (such as HOL 4, Isabelle/HOL, HOL Light, etc.), the Calculus of
constructions (such as Coq, Matita, Lean, etc.), and Simple type
theory with predicate subtyping and prenex polymorphism (such as PVS),
can be expressed.

Moreover, the proofs of the theory $\cU$ can be classified as proofs
in Minimal predicate logic, Constructive predicate logic, etc. just by
identifying the axioms they use, akin to proofs in geometry that can
be classified as proofs in Euclidean, hyperbolic, elliptic, neutral,
etc. geometries.  More precisely, we identify sub-theories of the
theory $\cU$ that correspond to each of these theories, and we prove
that when a proof in $\cU$ uses only symbols of a sub-theory,
then it is a proof in that sub-theory.

In Section~\ref{sec-lpmc}, we recall the definition of \lpr{} and of a
theory. In Section~\ref{sec-u}, we introduce the theory $\cU$ step by
step without any (other than informal) reference to a known system. In
Section~\ref{sec-frag-th}, we provide a general theorem on
sub-theories in \lpr{}, and prove that every fragment of $\cU$,
including $\cU$ itself, is indeed a theory, that is, it is defined by
a confluent and type-preserving rewriting system.  Finally, in
Section~\ref{sec-subth}, we detail the sub-theories of $\cU$ that
correspond to the above mentioned systems.

This paper is an extended version of \cite{blanqui21fscd}. It provides more
details on the encoding of various features (Section \ref{sec-u}), the
proof of the Fragment Theorem (Section \ref{sec-frag-th}), and the
examples of sub-theories (Section \ref{sec-subth}).

\section{The \lpcm{}}\label{sec-lpmc}

\lpr{} is an extension of the Edinburgh logical framework
\cite{HHP} with a primitive notion of computation defined with
rewriting rules \cite{dershowitz90chapter,terese03book}.

However, for defining terms, we use Barendregt's syntax for Pure Type
Systems \cite{barendregt92chapter}:
$$t,u = c~|~x~|~\TYPE~|~\KIND~|~\tpi{x}{t}{u}~|~\tabs{x}{t}{u}~|~t~u$$
where $c$ belongs to a finite or infinite set of constants $\cC$ and
$x$ to an infinite set $\cV$ of variables.
The terms $\TYPE$ and $\KIND$ are called sorts.
The term $\tpi{x}{t}{u}$ is called a product. It is dependent if the
variable $x$ occurs free in $u$. Otherwise, it is simply written $t\ra
u$. Terms are also often written $A$, $B$, etc.
The set of constants of a term $t$ is written $\const(t)$.

A rewriting rule is a pair of terms $\ell \lra r$, such that $\ell =
c~\ell_1~\ldots~\ell_n$, where $c$ is a constant.  If $\cR$ is a set of
rewriting rules, we write $\lra_\cR$ for the smallest relation
closed by term constructors and substitution containing $\cR$,
$\lra_\beta$ for the usual $\beta$-reduction, ${\lra_{\beta\cR}}$ for
${{\lra_\beta} \cup {\lra_\cR}}$, and $\equiv_{\beta\cR}$ for the
smallest equivalence relation containing ${\lra_{\beta\cR}}$.

A relation $\lra$ is confluent on a set of terms if, for all terms
$t,u,v$ in this set such that $t\lra^*u$ and $t\lra^*v$, there exists
a term $w$ in this set such that $u\lra^*w$ and $v\lra^*w$. Confluence
implies that every term has at most one irreducible form.

\begin{figure}[ht]
$$\irule{}
        {\vdash_{\Sigma,\cR} [~]~\mbox{well-formed}}
        {\mbox{(empty)}}$$
$$\irule{\Gamma \vdash_{\Sigma,\cR} A:s}
        {\vdash_{\Sigma,\cR} \Gamma, x:A~\mbox{well-formed}}
        {\mbox{(decl)}}$$
$$~$$
$$\irule{\vdash_{\Sigma,\cR} \Gamma~\mbox{well-formed}}
        {\Gamma \vdash_{\Sigma,\cR} \TYPE:\KIND}
        {\mbox{(sort)}}$$
$$\irule{\vdash_{\Sigma,\cR} \Gamma~\mbox{well-formed}~~~\vdash_{\Sigma,\cR} A:s}
        {\Gamma \vdash_{\Sigma,\cR} c:A}
        {\mbox{(const)}~~~c:A \in \Sigma}$$
$$\irule{\vdash_{\Sigma,\cR} \Gamma~\mbox{well-formed}}
        {\Gamma \vdash_{\Sigma,\cR} x:A}
        {\mbox{(var)}~~~x:A \in \Gamma}$$
$$\irule{\Gamma \vdash_{\Sigma,\cR} A:\TYPE & \Gamma,x:A \vdash_{\Sigma,\cR} B:s}
        {\Gamma \vdash_{\Sigma,\cR} \tpi{x}{A}{B}:s}
        {\mbox{(prod)}}$$
$$\irule{\Gamma\vdash_{\Sigma,\cR} A:\TYPE & \Gamma,x:A \vdash_{\Sigma,\cR} B:s & \Gamma,x:A\vdash_{\Sigma,\cR} t:B}
        {\Gamma\vdash_{\Sigma,\cR}\tabs{x}{A}{t} : \tpi{x}{A}{B}}
        {\mbox{(abs)}}$$
$$\irule{\Gamma \vdash_{\Sigma,\cR} t: \tpi{x}{A}{B} & \Gamma \vdash_{\Sigma,\cR} u:A}
        {\Gamma \vdash_{\Sigma,\cR} t~u:(u/x)B}
        {\mbox{(app)}}$$
$$\irule{\Gamma \vdash_{\Sigma,\cR} t:A & \Gamma \vdash_{\Sigma,\cR} B:s}
        {\Gamma \vdash_{\Sigma,\cR} t:B}
        {\mbox{(conv)}~~~A\equiv_{\beta\cR} B}$$
        \caption{Typing rules of \lpr{} with signature $\Sigma$ and rewriting rules
          $\mathcal{R}$\label{typingrules}}
\end{figure}

The typing rules of \lpr{} are given in Figure \ref{typingrules}. The
difference with the rules of the Edinburgh logical framework is that,
in the rule (conv), types are identified modulo $\equiv_{\beta\cR}$
instead of just $\equiv_\beta$.  In a typing judgement $\Gamma
\vdash_{\Sigma, \cR} t:A$, the term $t$ is given the type $A$ with
respect to three parameters: a signature $\Sigma$ that assigns a type
to the constants of $t$, a context $\Gamma$ that assigns a type to the
free variables of $t$, and a set of rewriting rules $\cR$.  A context
$\Gamma$ is a list of declarations $x_1:B_1,\ldots,x_m:B_m$ formed
with a variable and a term.  A signature $\Sigma$ is a list of
declarations $c_1:A_1,\ldots,c_n:A_n$ formed with a constant and a
closed term, that is a term with no free variables. This is why
the rule (const) requires no context for typing $A$. We write
$|\Sigma|$ for the set $\{c_1, \ldots, c_n\}$, and $\Lambda(\Sigma)$
for the set of terms $t$ such that $\const(t) \subseteq |\Sigma|$. We
say that a rewriting rule $\ell \lra r$ is in $\Lambda(\Sigma)$ if
$\ell$ and $r$ are, and a context $x_1:B_1, \ldots, x_m:B_m$ is in
$\Lambda(\Sigma)$ if $B_1, \ldots, B_m$ are.  It is often convenient
to group constant declarations and rules into small clusters, called
``axioms''.

A relation $\lra$ preserves typing in $(\Sigma,\cR)$ if, for all
contexts $\Gamma$ and terms $t$, $u$ and $A$ of $\Lambda(\Sigma)$, if
$\Gamma\vdash_{\Sigma,\cR} t:A$ and $t\lra u$, then
$\Gamma\vdash_{\Sigma,\cR} u:A$. The relation $\lra_\beta$ preserves
typing as soon as $\lra_{\beta\cR}$ is confluent (see for instance
\cite{blanqui01phd}) for, in this case, the product is injective
modulo $\equiv_{\beta\cR}$:
$\tpi{x}{A}{B}\equiv_{\beta\cR}\tpi{x}{A'}{B'}$ if and only if
$A\equiv_{\beta\cR}A'$ and $B\equiv_{\beta\cR}B'$.
The relation $\lra_\cR$ preserves typing if
every rewriting rule $\ell \lra r$ preserves typing, that is:
for all contexts $\Gamma$, substitutions $\theta$ and
terms $A$ of $\Lambda(\Sigma)$, if $\Gamma\vdash_{\Sigma,\cR} \theta
l:A$ then $\Gamma\vdash_{\Sigma,\cR} \theta r:A$.

Although typing is defined with arbitrary signatures $\Sigma$ and sets
of rewriting rules $\cR$, we are only interested in sets $\cR$
verifying some confluence and type-preservation properties.

\begin{defi}[System, theory]
A system is a pair $(\Sigma,\cR)$ such that each rule of $\cR$ is in
$\Lambda(\Sigma)$.  It is a theory if $\lra_{\beta\cR}$ is confluent
on $\Lambda(\Sigma)$, and every rule of $\cR$ preserves typing in
$(\Sigma,\cR)$.
\end{defi}

Therefore, in a theory, $\lra_{\beta\cR}$ preserves typing since
$\lra_\beta$ preserves typing (for $\lra_{\beta\cR}$ is confluent) and
$\lra_\cR$ preserves typing (for every rule preserves typing). We
recall two other basic properties of \lpr{} we will use in Theorem
\ref{th-frag}:

\begin{lem}
\label{typtypprod}
If $\Gamma\vdash_{\Sigma,\cR} t:A$, then either $A=\KIND$ or
  $\Gamma\vdash_{\Sigma,\cR} A:s$ for some sort $s$.

If $\Gamma\vdash_{\Sigma,\cR} \tpi{x}{A}{B}:s$, then
$\Gamma\vdash_{\Sigma,\cR} A:\TYPE$.
\end{lem}

\section{The theory $\cU$}\label{sec-u}

Let us now present the system $\cU$ which is formed with axioms expressed in
\lpr{}. We will prove in Theorem \ref{lem-frag-U} that this system is indeed a theory.

\subsection{Object-terms}

The notions of term, proposition, and proof are not primitive in
\lpr{}. The first axioms of the theory $\cU$ introduce these
notions. We first define a notion analogous to the Predicate logic
notion of term, to express the objects the theory speaks about, such
as the natural numbers.  As all expressions in \lpr{} are called
``terms'', we shall call these expressions ``object-terms'', to
distinguish them from the other terms.

To build the notion of object-term in \lpr{} we
declare a constant $\I$ of type $\TYPE$
\begin{leftbar}
$\I:\TYPE$ \hfill ($\I$-decl)
\end{leftbar}

\noindent
and constants of type $\I \ra
... \ra \I \ra \I$ for the function symbols, for instance a constant $0$
of type $\I$ and a constant \textit{succ} of type $\I \ra \I$.
The
object-terms, for instance $(\mbox{\textit{succ}}~(\mbox{\textit{succ}}~0))$ and
$(\mbox{\textit{succ}}~x)$, are then just \lpr{} terms of type $\I$
and, in an object-term, the variables are \lpr{} variables of
type $\I$.  If we wanted to have object-terms of several sorts, like in
Many-sorted predicate logic, we could just declare several constants
$I_1$, $I_2$, ..., $I_n$ of type $\TYPE$.

\subsection{Propositions}

Just like \lpr{} does not contain a primitive notion of object-term,
it does not contain a primitive notion of proposition, but tools to
define this notion. To do so, in the theory $\cU$, we declare a
constant $\Prop$ of type $\TYPE$

\begin{leftbar}
$\Prop: \TYPE$ \hfill ($\Prop$-decl)
\end{leftbar}

\noindent
and predicate symbols, that is constants of type $\I \ra ...  \ra \I
\ra \Prop$, for instance a constant $\pos$ of type $\I \ra
\Prop$. Propositions are then \lpr{} terms, such as
$(\mbox{\textit{positive}}~(\mbox{\textit{succ}}~(\mbox{\textit{succ}}~0)))$,
of type $\Prop$.

\subsection{Implication}

In the theory $\cU$, we then declare a constant for implication

\begin{leftbar}
$\imp : \Prop \ra \Prop \ra \Prop$\quad(written infix) \hfill ($\imp$-decl)
\end{leftbar}

\noindent
and we can then construct the term $(\pos~({\textit{succ}}~({\textit{succ}}~0))) \imp
(\pos~({\textit{succ}}~({\textit{succ}}~0)))$ of type $\Prop$.

\subsection{Proofs}

Predicate logic defines a language for terms and propositions,
but proofs have to be defined in a second step, for instance as
derivations in natural deduction, sequent calculus, etc. These
derivations, like object-terms and propositions, are trees. Therefore,
they can be represented as \lpr{} terms.

Using the Brouwer-Heyting-Kolmogorov interpretation, a proof of the
proposition $A \imp B$ should be a \lpr{} term expressing a function
mapping proofs of $A$ to proofs of $B$.  Then, using the Curry-de
Bruijn-Howard correspondence, the type of this term should be the
proposition $A \imp B$ itself. But, this is not possible in the theory
$\cU$ yet, as the proposition $A \imp B$ has the type $\Prop$, and not
the type $\TYPE$.

A strict view on the Curry-de Bruijn-Howard correspondence leads to
identify $\Prop$ and $\TYPE$, yielding the conclusion that all types,
including $\I$, and $\Prop$ itself, are propositions. A more moderate
view leads to the introduction of an embedding $\Prf$ of propositions into
types, mapping each proposition $A$ to the type $\Prf~A$ of its
proofs.  This view is moderate in two respects. First, the
propositions themselves are not the types of their proofs: if $t$ is a
proof of $A$, then it does not have the type $A$, but the type
$\Prf~A$. Second, this embedding is not surjective. So not all types
are types of proofs, in particular $\I$ and $\Prop$ are not.

So, in the theory $\cU$, we declare a constant $\Prf$

\begin{leftbar}
$\Prf:\Prop \ra \TYPE$ \hfill ($\Prf$-decl)
\end{leftbar}

\noindent
When assigning the type $\Prop \ra \TYPE$ to the constant $\Prf$, we
use the fact that \lpr{} supports dependent types, that is the
possibility to build a family of types $\Prf~x$ parameterized with
a variable $x$ of type $\Prop$, where $\Prop$ is itself of type $\TYPE$.

According to the Brouwer-Heyting-Kolmogorov interpretation, a proof of
$A \imp A$ is a \lpr{} term expressing a function mapping proofs of
$A$ to proofs of $A$, so that it can be both built and used as a
function. In particular, the identity function $\tabs{x}{\Prf~A}{x}$
mapping each proof of $A$ to itself is a proof of $A \imp A$.
According to the Curry-de Bruijn-Howard correspondence, this term
should have the type $\Prf~(A \imp A)$, but it has the type $\Prf~A
\ra \Prf~A$.  So, the types $\Prf~(A \imp A)$ and $\Prf~A \ra \Prf~A$
must be identified.  To do so, we use the fact that \lpr{} allows the
declaration of rewriting rules, so that $\Prf~(A \imp A)$ rewrites to
$\Prf~A \ra \Prf~A$.

\begin{leftbar}
$\Prf~(x \imp y) ~\lra~ \Prf~x \ra \Prf~y$ \hfill ($\imp$-red)
\end{leftbar}

This rule expresses the meaning of the constant $\imp$. It is, in
\lpr{}, the expression of the Brouwer-Heyting-Kolmogorov
interpretation of proofs for implication: a proof of $x \imp y$ is a
function mapping proofs of $x$ to proofs of $y$.  So, in the theory
$\cU$, the Brouwer-Heyting-Kolmogorov interpretation of proofs for
implication is made explicit: it is the rule ($\imp$-red).

\subsection{Universal quantification}

Unlike implication, the universal quantifier binds a
variable. Thus, we express the proposition $\forall z~A$ as the
proposition $\forall~(\tabs{z}{\I}{A})$
\cite{Church40,NadathurMiller,Paulson,HHP}, yielding the type
  $\I \ra \Prop$ for the argument of $\forall$, hence the
type $(\I \ra \Prop) \ra \Prop$ for the constant
$\forall$ itself.

But, in the theory $\cU$, we allow quantification, not only over the
variables of type $\I$, but over variables of any type of
object-terms. We could introduce a different quantifier for each type
of object-terms, for instance two quantifiers of type $(I_1 \ra \Prop)
\ra \Prop$ and $(I_2 \ra \Prop) \ra \Prop$ if we had two types $I_1$
and $I_2$ of object-terms. But, as in some cases, we will have an infinite
number of types of object-terms, this would require the introduction of an
infinite number of constants.

Thus, we rather want to have a single generic quantifier. But we
cannot give the type $\tpi{X}{\TYPE}{(X \ra \Prop) \ra \Prop}$ to this
quantifier, first because in \lpr{} there is no way to quantifiy over
a variable of type $\TYPE$, but also because this would introduce the
possibility to quantify over $\Prop$ and all the types of the form
$\Prf~A$, while we do not always want to consider these types as
types of object-terms.

Therefore, in
  the theory $\cU$, we declare a constant $\Set$ of
  type $\TYPE$ for the types of object-terms

\begin{leftbar}
$\Set:\TYPE$ \hfill ($\Set$-decl)
\end{leftbar}

\noindent
a constant $\bliota$ of type $\Set$

\begin{leftbar}
$\bliota:\Set$ \hfill ($\bliota$-decl)
\end{leftbar}

\noindent
a constant $\El$ to embed the terms of type $\Set$ into terms
of type $\TYPE$

\begin{leftbar}
$\El:\Set \ra \TYPE$ \hfill ($\El$-decl)
\end{leftbar}

\noindent

and a rule that reduces the term $\El~\bliota$ to $\I$
\begin{leftbar}
$\El~\bliota ~\lra~ \I$ \hfill ($\bliota$-red)
\end{leftbar}

\noindent
The types of object-terms then have the form $\El~A$ and are
distinguished among the other terms of type $\TYPE$.

We can now give the type $\tpi{x}{\Set}{(\El~x \ra \Prop) \ra \Prop}$ to the
generic universal quantifier
\begin{leftbar}
$\fa:\tpi{x}{\Set}{(\El~x \ra \Prop) \ra \Prop}$ \hfill ($\fa$-decl)
\end{leftbar}

\noindent
and write $\fa~\bliota~(\tabs{z}{\I}{A})$ for the proposition $\forall
z~A$.

Just like for the implication, we declare a rewriting rule expressing
that the type of the proofs of the proposition $\fa~x~p$ is the type
of functions mapping each $z$ of type $\El~x$ to a proof of $p~z$

\begin{leftbar}
$\Prf~(\fa~x~p) ~\lra~ \tpi{z}{\El~x}{\Prf~(p~z)}$ \hfill ($\fa$-red)
\end{leftbar}

Again, the Brouwer-Heyting-Kolmogorov interpretation of proofs for the
universal quantifier is made explicit: it is this rule ($\fa$-red).

\subsection{Other constructive connectives and quantifiers}

The other connectives and quantifiers are defined \textit{\`a la} Russell.
For the conjunction, for example,
$\Prf~(x \conj y)$ is defined as 
$\tpi{z}{\Prop}{(\Prf~x \ra \Prf~y \ra \Prf~z) \ra \Prf~z}$.
This definition does not use the quantifier $\fa$ of the theory
$\cU$ (so far, in the theory $\cU$, we can
quantify over the type $\I$, but not over the type $\Prop$),
but the quantifier $\Pi$ of the logical framework \lpr{} itself.

Remark that, \textit{per se}, the quantification on the variable $z$ of
type $\Prop$ is predicative, as the term $\tpi{z}{\Prop}{(\Prf~x \ra
  \Prf~y \ra \Prf~z) \ra \Prf~z}$ has type $\TYPE$ and not $\Prop$.
But, the rule rewriting $\Prf~(x \conj y)$ to $\tpi{z}{\Prop}{(\Prf~x
  \ra \Prf~y \ra \Prf~z) \ra \Prf~z}$ introduces some impredicativity,
as the term $x \conj y$ of type $\Prop$ is ``defined'' as the inverse image,
for the embedding $\Prf$, of the type $\tpi{z}{\Prop}{(\Prf~x \ra
  \Prf~y \ra \Prf~z) \ra \Prf~z}$, that contains a quantification on a
variable of type $\Prop$

\begin{leftbar}
$\bltop:\Prop$ \hfill ($\bltop$-decl)

$\Prf~\bltop ~\lra~ \tpi{z}{\Prop}{\Prf~z \ra \Prf~z}$ \hfill ($\bltop$-red)

\smallskip

$\blbot:\Prop$ \hfill ($\blbot$-decl)

$\Prf~\blbot ~\lra~ \tpi{z}{\Prop}{\Prf~z}$ \hfill ($\blbot$-red)

\smallskip

$\blneg:\Prop \ra \Prop$ \hfill ($\blneg$-decl)

$\Prf~(\blneg x) ~\lra~ \Prf~x \ra \tpi{z}{\Prop}{\Prf~z}$ \hfill ($\blneg$-red)

\smallskip

$\conj:\Prop \ra \Prop \ra \Prop$ \quad(written infix)\hfill ($\conj$-decl)

$\Prf~(x \conj y)~\lra~ \tpi{z}{\Prop}{(\Prf~x \ra \Prf~y \ra \Prf~z)
  \ra \Prf~z}$ \hfill ($\conj$-red)

\smallskip

$\disj:\Prop \ra \Prop \ra \Prop$ \quad(written infix)\hfill
($\disj$-decl)

$\Prf~(x \disj y)~\lra~ \tpi{z}{\Prop}{(\Prf~x \ra \Prf~z) \ra (\Prf~y
  \ra \Prf~z) \ra \Prf~z}$\hfill ($\disj$-red)

\smallskip

$\ex:\tpi{a}{\Set}{(\El~a \ra \Prop) \ra \Prop}$ \hfill ($\ex$-decl)

$\Prf~(\ex~a~p)~\lra~ \tpi{z}{\Prop}{(\tpi{x}{\El~a}{\Prf~(p~x) \ra
    \Prf~z}) \ra \Prf~z}$\hfill ($\ex$-red)

\end{leftbar}

\subsection{Infinity}

Now that we have the symbols $\bltop$ and $\blbot$, we can express
that the type $\I$ is infinite, that is, that there
exists a non-surjective injection from this type to itself. We
call this non-surjective injection $\S$. To express its injectivity,
we introduce its left inverse $\pred$. To express its non-surjectivity, we introduce an element $\0$, that is not in its image
$\positive$ \cite{DowekWernerPA}.
This choice of notation enables the definition of natural numbers as some
elements of type $\I$

\begin{leftbar}
  $\0:\I$ \hfill ($\0$-decl)

\smallskip

$\S:\I \ra \I$ \hfill ($\S$-decl)

\smallskip

$\pred:\I \ra \I$ \hfill ($\pred$-decl)

$\pred~\0 ~\lra~ \0$ \hfill ($\pred$-red1)

$\pred~(\S~x) ~\lra~ x$ \hfill ($\pred$-red2)

\smallskip

$\positive:\I \ra \Prop$ \hfill ($\positive$-decl)

$\positive~\0 ~\lra~ \blbot$ \hfill ($\positive$-red1)

$\positive~(\S~x) ~\lra~ \bltop$ \hfill ($\positive$-red2)
\end{leftbar}

\subsection{Classical connectives and quantifiers}

The disjunction in constructive logic and in classical logic are
governed by different deduction rules. As the deduction rules express
the meaning of the connectives and quantifiers, we can conclude that the
disjunction in constructive logic and in classical logic have
different meanings.  If these disjunctions have different meanings,
they should be expressed with different symbols, for instance $\vee$ for
the constructive disjunction and $\vee_c$ for the classical one, just
like, in classical logic, we use two different symbols for the
inclusive disjunction and the exclusive one.

The constructive and the classical disjunction need not belong to
different languages, but they can coexist in the same.  Ecumenical
logics \cite{Prawitz15,Dowek15,PereiraRodriguez17,Grienenberger19}
are logics where
the constructive and classical connectives and quantifiers coexist.  A
proposition whose connectors and quantifiers are all constructive, is
said to be ``purely constructive'', and one whose connectors and
quantifiers are all classical, is said to be ``purely classical''.
The others are said to be ``mixed propositions''.  Any deductive
system, where a purely constructive proposition is provable if and
only if it is provable in Constructive predicate logic, and where a
purely classical proposition is provable if and only if it is provable
in Classical predicate logic, is Ecumenical. Ecumenical logics may of
course differ on mixed propositions.

Many Ecumenical logics consider the constructive connectives and
quantifiers as primitive and attempt to define the classical ones from
them, using the negative translation as a definition. There are
several options: the classical disjunction, for instance, can
be defined in any of the following ways:
\begin{enumerate}
\item $A \vee_c B = \neg \neg (A \vee B)$
\item $A \vee_c B = (\neg \neg A) \vee (\neg \neg B)$
\item $A \vee_c B = \neg \neg ((\neg \neg A) \vee (\neg \neg B))$
\end{enumerate}
and similarly for the other connectives and quantifiers.

Using these definitions, the proposition $(P \wedge_c Q) \Rightarrow_c P$
is then:
\begin{enumerate}
\item $\neg \neg ((\neg \neg (P \wedge Q)) \Rightarrow P)$
\item $(\neg \neg ((\neg \neg P) \wedge (\neg \neg Q)))
\Rightarrow (\neg \neg P)$
\item  $\neg \neg ((\neg \neg \neg \neg ((\neg \neg P) \wedge (\neg \neg Q)))
  \Rightarrow (\neg \neg P))$
\end{enumerate}
None of them is exactly the negative translation of $(P \wedge Q)
\Rightarrow P$ that is
$$\neg \neg ((\neg \neg ((\neg \neg P) \wedge (\neg \neg Q)))
\Rightarrow (\neg \neg P))$$

With Definition (1), the double negations on atomic propositions are
missing.
This can be repaired in two ways.
Predicate symbols of the language can be duplicated~\cite{Prawitz15}
into a constructive and a classical counterpart, the latter being the
the double negation of the former.
Or the syntax of predicate logic can be modified~\cite{Gilbert}.
First, terms are defined.
Then,
atoms are defined as terms of the form $P(t_1, \ldots, t_n)$
where $P$ is a predicate symbol and $t_1, \ldots, t_n$ are terms.
Finally, propositions are defined as either explicitly embedded atoms,
conjunctions of two propositions, etc.
Atoms can be constructively embedded into propositions
with the symbol $\rhd$ or classically with a
double-negation version of $\rhd$.
This way, the proposition above is now written
$(\rhd_c P \wedge_c~\rhd_c Q) \Rightarrow_c \rhd_c P$, which is, by
definition, equal to $\neg \neg ((\neg \neg ((\neg \neg \rhd P) \wedge
(\neg \neg \rhd Q))) \Rightarrow (\neg \neg \rhd P))$.

With Definition (2) \cite{AllaliHermant}, the double negation at the
root of the proposition is missing.  This again can be repaired
\cite{Grienenberger19}
by modifying the syntax of Predicate logic,
defining first terms, then pre-propositions, that are defined like the
propositions in Predicate logic and then propositions, a proposition
being obtained by applying a symbol $\circ_c$ to a pre-proposition.
Again, this symbol has also a classical version defined as the double
negation of the constructive one. This way, the proposition above is
written $\circ_c ((P \wedge_c Q) \Rightarrow_c P)$ and this proposition
is, by definition, equal to $\neg \neg ((\neg \neg ((\neg \neg P)
\wedge (\neg \neg Q))) \Rightarrow ( \neg \neg P))$.

Definition (3) \cite{Dowek15} is closer to the negative translation
except that, in some places, the two negations are replaced with
four. But, as $\neg \neg \neg A$ is equivalent to $\neg A$, these
extra negations can be removed. Yet, a classical atomic proposition
$P$ is the same as its constructive version, while its negative
translation is $\neg \neg P$, and in (1) $P_c$ or $\rhd_c P$ is
equal to $\neg \neg P$, as well as $\circ_c P$ in (2).  As atomic
propositions are not provable anyway, this does not affect
provability. But it affects hypothetical provability, leading to
duplicate the notion of entailment.

In the theory $\cU$, we use Definition (2). Indeed, as we
already have a distinction between the proposition $A$ and the type
$\Prf~A$ of its proofs, we can just include the symbol $\circ_{c}$
into the constant $\Prf$ introducing a classical version
$\Prfc$ of this constant

\begin{leftbar}
$\Prfc:\Prop \ra \TYPE$ \hfill ($\Prfc$-decl)

$\Prfc ~\lra~ \tabs{x}{\Prop}{\Prf~(\blneg \blneg x)}$ \hfill ($\Prfc$-red)
\end{leftbar}

We can then define the classical connectives and quantifiers as follows

\begin{leftbar}
$\impc:\Prop \ra \Prop \ra \Prop$ \quad(written infix)\hfill ($\impc$-decl)

$\impc ~\lra~ \tabs{x}{\Prop}{\tabs{y}{\Prop}{(\blneg \blneg x) \imp
      (\blneg \blneg y)}}$ \hfill ($\impc$-red)

\smallskip

$\conjc:\Prop \ra \Prop \ra \Prop$ \quad(written infix)\hfill
($\conjc$-decl)

$\conjc ~\lra~ \tabs{x}{\Prop}{\tabs{y}{\Prop}{(\blneg \blneg x) \conj
    (\blneg \blneg y)}}$ \hfill ($\conjc$-red)

\smallskip

$\disjc:\Prop \ra \Prop \ra \Prop$ \quad(written infix)\hfill
($\disjc$-decl)

$\disjc ~\lra~ \tabs{x}{\Prop}{\tabs{y}{\Prop}{(\blneg \blneg x) \disj
    (\blneg \blneg y)}}$ \hfill ($\disjc$-red)

\smallskip

$\fac : \tpi{a}{\Set}{(\El~a \ra \Prop) \ra \Prop}$ \hfill
($\fac$-decl)

$\fac ~\lra~ \tabs{a}{\Set}{\tabs{p}{(\El~a \ra
    \Prop)}{\fa~a~(\tabs{x}{\El~a}{\blneg \blneg (p~x))}}}$ \hfill
($\fac$-red)

\smallskip

$\exc : \tpi{a}{\Set}{(\El~a \ra \Prop) \ra \Prop}$ \hfill
($\exc$-decl)

$\exc ~\lra~ \tabs{a}{\Set}{\tabs{p}{(\El~a \ra
    \Prop)}{\ex~a~(\tabs{x}{\El~a}{\blneg \blneg (p~x))}}}$ \hfill
($\exc$-red)
\end{leftbar}

Note that $\top_c$ and $\bot_c$ are $\top$ and $\bot$, by definition.
Note also that $\neg \neg \neg A$ is equivalent to $\neg A$, so
we do not need to duplicate negation either.

\subsection{Propositions as objects}

So far, we have mainly reconstructed the Predicate logic notions
of object-term, proposition, and proof.  We can now turn to two
notions coming from Simple type theory: propositions as objects and
functionality.

Simple type theory is often presented as an independent system, but it
can be expressed in several logical frameworks, such as Predicate
logic, Isabelle, Deduction modulo theory, Pure
type systems, and also \lpr{}. Yet, the relation between
Predicate logic and Simple type theory is complex because
\begin{itemize}
\item Simple type theory can be expressed in Predicate logic
\item and Predicate logic is a restriction of Simple type theory, allowing
quantification on variables of type $\iota$ only.
\end{itemize}

So, in Predicate logic, we can express Simple type theory, that
contains, as a restriction, Predicate logic, in which we can express
Simple type theory, that contains, as a restriction, Predicate logic, in
which we can express Simple type theory, that contains, etc. Stacking
encodings in this way leads to nonsensical expressions of Simple type
theory. But this remark shows that, after having reconstructed
Predicate logic in \lpr{}, we have a choice: we can either express
Simple type theory in Predicate logic, that is itself expressed in
\lpr{}, or express Simple type theory directly in \lpr{}, letting
Predicate logic be, a posteriori, a restriction of it, that is, build
Simple type theory, not in Predicate logic, but as an extension of
Predicate logic.

In the theory $\cU$, we choose the second option that leads to a
simpler expression of Simple type theory, avoiding the stacking of two
encodings.  Simple type theory is thus expressed by adding two axioms
on top of Predicate logic: one for propositions as objects and one for
functionality.

Let us start with propositions as objects.  So far, the term $\bliota$
is the only closed term of type $\Set$.  So, we can only quantify over
the variables of type $\El~\bliota$, that is $\I$.  In particular, we cannot
quantify over propositions.  To do so, we just need to declare a constant
$\o$ of type $\Set$

\begin{leftbar}
$\o:\Set$ \hfill ($\o$-decl)
\end{leftbar}

\noindent
and a rule identifying $\El~\o$ and $\Prop$

\begin{leftbar}
$\El~\o ~\lra~ \Prop$ \hfill ($\o$-red)
\end{leftbar}

Note that just like there are no terms of type $\bliota$, but terms,
such as $\0$, which have type $\El~\bliota$, that is $\I$, there are
no terms of type $\o$, but terms, such as $\bltop$, which have type
$\El~\o$, that is $\Prop$.

Applying the constant $\fa$ to the constant $\o$, we obtain a term of
type $(\El~\o \ra \Prop) \ra \Prop$, that is $(\Prop \ra \Prop) \ra
\Prop$, and we can express the proposition $\forall p~(p \Rightarrow
p)$ as $\fa~\o~(\tabs{p}{\Prop}{p \imp p})$.  The type
$\Prf~(\fa~\o~(\tabs{p}{\Prop}{p \imp p}))$ of the proofs of this
proposition rewrites to $\tpi{p}{\Prop}{\Prf~p \ra \Prf~p}$. So, the
term $\tabs{p}{\Prop}{\tabs{x}{\Prf~p}{x}}$ is a proof of this
proposition.

\subsection{Functionality}

Besides $\bliota$ and $\o$, we introduce more types in the theory, for
functions and sets. To do so, we declare a constant
\begin{leftbar}
$\arr:\Set \ra \Set \ra \Set$ \quad(written infix)\hfill ($\arr$-decl)
\end{leftbar}

\noindent
and a rewriting rule

\begin{leftbar}
$\El~(x \arr y) ~\lra~ \El~x \ra \El~y$ \hfill ($\arr$-red)
\end{leftbar}

For instance, these rules enable the construction of the \lpr{} term $\bliota
\arr \bliota$ of type $\Set$ that expresses the simple type $\iota \ra
\iota$. The \lpr{} term $\El~(\bliota \arr \bliota)$ of type $\TYPE$
rewrites to $\I \ra \I$. The simply typed term
$\tabs{x}{\iota}{x}$ of type $\iota \ra \iota$ is then expressed as
the term $\tabs{x}{\I}{x}$ of type $\I \ra
\I$ that is, $\El~(\bliota \arr \bliota)$.

\subsection{Dependent arrow}\label{sec:dependent-function-types}

The axiom ($\arr$) enables us to give simple types to the object-terms
expressing functions. We can also give them dependent types, with the
dependent versions of this axiom

\begin{leftbar}
$\arrd : \tpi{x}{\Set}{(\El~x \ra \Set) \ra \Set}$\quad(written infix)
  \hfill ($\arrd$-decl)

$\El~(x \arrd y) ~\lra~ \tpi{z}{\El~x}{\El~(y~z)}$ \hfill
  ($\arrd$-red)
\end{leftbar}

Note that, if we apply the constant $\arrd$ to a term $t$ and a term
$\tabs{z}{\El~t}{u}$, where the variable $z$ does not occur in $u$,
then $\El~(t \arrd \tabs{z}{\El~t}{u})$ rewrites to $\El~t \ra \El~u$,
just like $\El~(t \arr u)$. Thus, the constant $\arrd$ is useful only
if we can build a term $\tabs{z}{\El~t}{u}$ where the variable $z$
occurs in $u$. With the symbols we have introduced so far, this is not
possible. The only constants that can be used to build a term of type
$\Set$ are $\bliota$, $\o$, $\arr$, and $\arrd$, and the variable $z$
cannot occur free in a term built from these four constants and a
variable $z$ of type $\El~t$.

Just like we have a constant $\bliota$ of type $\Set$, we could add a
constant \textit{array} of type $\I \ra \Set$ such that
$\mbox{\textit{array}}~n$ is the type of arrays of length $n$.  We could
then construct the term $(\bliota \arrd
\tabs{n}{\I}{\mbox{\textit{array}}~n})$ of type $\Set$.
Then, the type $\El~(\bliota \arrd
  \tabs{n}{\I}{\mbox{\textit{array}}~n})$ that rewrites to
$\tpi{n}{\I}{\El~(\mbox{\textit{array}}~n)}$, would be the type of
functions mapping a natural number $n$ to an array of length $n$.

So, this symbol $\arrd$ becomes useful, only if we add such a
constant \textit{array}, object-level dependent types, or the symbols
$\blpi$ or $\psub$ below.

\subsection{Dependent implication}

In the same way, we can add a dependent implication, where, in the
proposition $A \Rightarrow B$, the proof of $A$ may occur in $B$

\begin{leftbar}
$\impd : \tpi{x}{\Prop}{(\Prf~x \ra \Prop) \ra \Prop}$ \quad(written
  infix) \hfill ($\impd$-decl)

$\Prf~(x \impd y)~\lra~\tpi{z}{\Prf~x}{\Prf~(y~z)}$ \hfill
  ($\impd$-red)
\end{leftbar}

\subsection{Proofs in object-terms}

To construct an object-term, we sometimes want to apply a function
symbol to other object-terms and also to proofs. For instance, we may
want to apply the Euclidean division {\textit{div}} to two numbers $t$ and
$u$ and to a proof that $u$ is positive.

We would like the type of {\textit{div}} to be something like
$$\El~(\bliota \arr \bliota \arrd \tabs{y}{\I}~{(\positive~y \arr \bliota)})$$
But the term $(\positive~y \arr \bliota)$ is not well typed, as the
constant $\arr$ expects, as a first argument, a term of type $\Set$
and not of type $\Prop$, that is, a type of object-terms and not of proofs.

Thus, we must declare another constant

\smallskip

\noindent
$\blpi:\Prop \ra \Set \ra \Set$

\smallskip

\noindent
and a rewriting rule

\smallskip

\noindent
$\El~(\blpi~x~y) ~\lra~ (\Prf~x) \ra (\El~y)$

\smallskip

\noindent
Just like for the constant $\arr$ and $\imp$, we can also have a
dependent version of this constant. In fact, in the theory $\cU$,
we only have this dependent version

\begin{leftbar}
$\blpi : \tpi{x}{\Prop}{(\Prf~x \ra \Set) \ra \Set}$ \hfill
  ($\blpi$-decl)

$\El~(\blpi~x~y)~\lra~\tpi{z}{\Prf~x}{\El~(y~z)}$ \hfill ($\blpi$-red)
\end{leftbar}

\noindent
This way, we can give, to the constant {\textit{div}}, the type
$$\El~(\bliota \arr \bliota \arrd
\tabs{y}{\I}{\blpi~(\positive~y)~(\tabs{z}{\Prf~(\positive~y)}{\bliota})})$$
that is
$$\I \ra \tpi{y}{\I}{\Prf~(\positive~y) \ra \I}$$

In the same way, if we add a symbol \(=\) of type
\(\tpi {x}{\Set}{\tapp{\El}{x} \ra \tapp{\El}{x} \ra \Prop}\), we can express
the proposition
$$\positive~y \impd
\tabs{p}{\Prf~(\positive~y)}{(=~\bliota~(\mbox{\textit{div}}~x~y~p)~(\mbox{\textit{div}}~x~y~p))}$$
enlightening the meaning of the proposition
  usually written
$$y > 0 \Rightarrow x/y = x/y$$

The proposition $x/y = x/y$ is well-formed, but it contains, besides
$x$ and $y$, an implicit free variable $p$, for a proof of $y >
0$. This variable is bound by the implication, that needs therefore to
be a dependent implication. Hence, the only free variables in
$y > 0 \Rightarrow x/y = x/y$ are $x$ and $y$.

\subsection{Proof irrelevance}\label{sec:proof-irrelevance}

If $p$ and $q$ are two non convertible proofs of the proposition
$\positive~2$, the terms $\mbox{\textit{div}}~7~2~p$ and
$\mbox{\textit{div}}~7~2~q$ are not convertible.
As a consequence, the proposition
$$=~\bliota~(\mbox{\textit{div}}~7~2~p)~(\mbox{\textit{div}}~7~2~q)$$
would not be provable.

To make these terms convertible, we embed the theory into an extended
one, that contains another constant
$$\mbox{\textit{div}}^{\dagger}:\El~(\bliota \arr \bliota \arr \bliota)$$
and a rule
$$\mbox{\textit{div}}~x~y~p ~\lra~ \mbox{\textit{div}}^{\dagger}~x~y$$ and we
define convertibility in this extended theory.  This way, the terms
$\mbox{\textit{div}}~7~2~p$ and $\mbox{\textit{div}}~7~2~q$ are convertible,
as they both reduce to $\mbox{\textit{div}}^{\dagger}~7~2$.

Note that, in the extended theory, the constant $\mbox{\textit{div}}^{\dagger}$ enables the construction of the erroneous term $\mbox{\textit{div}}^{\dagger}~1~0$.  But the extended theory is only used to
define the convertibility in the restricted one and this term is not a
term of the restricted theory. It is not even the reduct of a term of
the form $\mbox{\textit{div}}~1~0~r$ \cite{ferey19irrelevance,blanqui20types}.

\subsection{Dependent pairs and predicate subtyping}

Instead of declaring a constant \textit{div} that takes three
arguments: a number $t$, a number $u$, and a proof $p$ that $u$ is
positive, we can declare a constant that takes two arguments: a number
$t$ and a pair $\pair~\bliota~\positive~u~p$ formed with a number $u$
and a proof $p$ that $u$ is positive.

The type of the pair $\pair~\bliota~\positive~u~p$ whose first element
is a number and the second a proof that this number is positive is
written $\psub~\bliota~\positive$, or informally $\{x:\bliota \mid
\positive~x\}$. It can be called ``the type of positive numbers'',
especially if the pair is proof-irrelevant in its second argument. It
is a subtype of the type of natural numbers defined with the predicate
$\positive$. Therefore, the symbol $\psub$ introduces predicate
subtyping.

We thus declare a constant $\psub$ and a constant $\pair$

\begin{leftbar}
$\psub: \tpi{t}{\Set}{(\El~t \ra \Prop) \ra \Set}$ \hfill ($\psub$-decl)

\smallskip

$\pair: \tpi{t}{\Set}{\tpi{p}{\El~t \ra \Prop}{\tpi{m}{\El~t}{\Prf~(p~m) \ra \El~(\psub~t~p)}}}$\hfill ($\pair$-decl)
\end{leftbar}

This way, instead of giving the type $\El~(\bliota \arr \bliota \arrd
\tabs{y}{\Prf~(\positive~y)}{\bliota})$ to the constant \textit{div},
we can give it the type $\El~(\bliota \arr \psub~\bliota~\positive
\arr \bliota)$.

To avoid introducing a new positive number
$\pair~\bliota~\positive~3~p$ with each proof $p$ that $3$ is
positive, we make this symbol $\pair$ proof
irrelevant by introducing a
symbol $\paird$ and a rewriting rule that discards the proof

\begin{leftbar}
$\paird: \tpi{t}{\Set}{\tpi{p}{\El~t\ra\Prop}{\El~t \ra
      \El~(\psub~t~p)}}$ \hfill ($\paird$-decl)

\smallskip

$\pair~t~p~m~h ~\lra~ \paird~t~p~m$ \hfill ($\pair$-red)
\end{leftbar}

This declaration and this rewriting rule are not part of the theory $\cU$
but of the theory $\cU^{\dagger}$ used to define the conversion on the
terms of $\cU$.

Finally, we declare the projections $\fst$ and $\snd$ together with an
associated rewriting rule

\begin{leftbar}
$\fst: \tpi{t}{\Set}{\tpi{p}{\El~t \ra \Prop}{\El~(\psub~t~p) \ra
      \El~t}}$ \hfill ($\fst$-decl)

$\fst~t~p~(\paird~t'~p'~m) ~\lra~ m$ \hfill ($\fst$-red)

\smallskip

$\snd: \tpi{t}{\Set}{\tpi{p}{\El~t \ra
    \Prop}{\tpi{m}{\El~(\psub~t~p)}{\Prf~(p~(\fst~t~p~m))}}}$\hfill
($\snd$-decl)
\end{leftbar}

Note that the left hand side of the rule ($\fst$-red) is not
well-typed, but it can match a well-typed term
$\fst~A~B~(\paird~A~B~m)$. Yet, we prefer this rule to the non
linear one $\fst~t~p~(\paird~t~p~m) ~\lra~ m$ that would make
confluence proofs more difficult.

Note that there is no rewriting rule for the second projection as the
second element of pairs is discarded during rewriting.

\subsection{Dependent types}

When we have the axioms ($\El\/$), ($\bliota$), and ($\arr$), the type
$\I \ra \I$ of the term $\S$ is equivalent to the type $\El~(\bliota
\arr \bliota)$. Hence, this type $\I \ra \I$ is in the image of the
embedding $\El$.
The symbol \textit{array} introduced in Section~\ref{sec:dependent-function-types}
has the type $\I \ra \Set$ and similarly
if we have a predicate symbol $\leq: \El~ (\bliota~
\arr~ \bliota~ \arr \o)$, the term $\tabs{n}{\I}
{\psub~ \bliota~ (\tabs{m}{\I}{m \leq n})}$ also has the type $\I \ra \Set$.
Unlike the type $\I \ra \I$ the type \(\I \ra \Set\) is not in the image of any
embedding. It is well-formed in the framework but not in the theory itself.

To make this type an element of the image of an embedding we can introduce
dependent types at the level of objects-terms.
To do so, we introduce  a constant $\Sett$ of type
$\TYPE$ and a constant $\set$ of type
$\Sett$
\begin{leftbar}
  $\Sett: \TYPE$ \hfill ($\Sett$-decl)

  $\set: \Sett \hfill$ \hfill ($\set$-decl)
\end{leftbar}
\noindent a new arrow
\begin{leftbar}
  $\dtarr: \tpi{x}{\Set}{(\El~x \ra \Sett) \ra \Sett}$ \hfill ($\dtarr$-decl)
\end{leftbar}
\noindent and an embedding of $\Sett$ into $\TYPE$, similar to the
embeddings $\Prf$ and $\El$
\begin{leftbar}
  $\Ty: \Sett \ra \TYPE$ \hfill ($\Ty$-decl)
\end{leftbar}
\noindent and we identify $\Ty~ \set$ with $\Set$---like $\El~ \o$ is
identified with $\Prop$---and the dependent arrow $\dtarr$ with a
product type
\begin{leftbar}
  $\Ty~ \set \lra \Set$ \hfill ($\set$-red)

  $\Ty~ (x \dtarr y) \lra \tpi{z}{\El~x} \Ty~ (y~z)$ \hfill ($\dtarr$-red)
\end{leftbar}
\noindent The type $\I \ra \Set$, is now equivalent to
$\Ty~ (\bliota \dtarr (\tabs{n}{\I}{\set}))$
and is thus in the image of the embedding $\Ty$.
One could think in simply taking \(\set:\Set\), saving constants \(\Sett\),
\(\Ty\) and \(\dtarr\) and their rewrite rules.
However, such a declaration would encode the product
\((\triangle, \square, \square)\)
of system \(\lambda\mathrm{U}^-\)
which is inconsistent \cite{coquand:inria00076023,10.1007/BFb0014058}.

\subsection{Prenex predicative type quantification in types}

Using the symbols of the theory $\cU$ introduced so far, the
symbol for equality $=$ has the type
$\tpi{x}{\Set}{\tapp{\El}{x} \ra \tapp{\El}{x} \ra \Prop}$
which is not a type of object terms.
This motivates the
introduction of object-level polymorphism
\cite{girard1972interpretation,reynolds1974towards}.  However
extending Simple type theory with object-level polymorphism makes it
inconsistent \cite{10.1007/BFb0014058,coquand:inria00076023}, and
similarly it makes the theory $\cU$ inconsistent.  So, object-level
polymorphism in $\cU$ is restricted to prenex polymorphism.

To do so, we introduce a new constant $\Scheme$ of type $\TYPE$

\begin{leftbar}
$\Scheme:\TYPE$ \hfill ($\Scheme$-decl)
\end{leftbar}

\noindent
a constant $\Els$ to embed the terms of type $\Scheme$ into terms of
type $\TYPE$

\begin{leftbar}
$\Els:\Scheme \ra \TYPE$ \hfill ($\Els$-decl)
\end{leftbar}

\noindent
a constant $\lift$ to embed the terms of type $\Set$ into terms of
type $\Scheme$ and a rule connecting these embeddings

\begin{leftbar}
$\lift:\Set \ra \Scheme$ \hfill ($\lift$-decl)

$\Els~(\lift~x) ~\lra~ \El~x$ \hfill ($\lift$-red)
\end{leftbar}

\noindent
We then introduce a quantifier for the variables of type $\Set$ in the
terms of type $\Scheme$ and the associated rewriting rule

\begin{leftbar}
$\calfa:(\Set \ra \Scheme) \ra \Scheme$ \hfill ($\calfa$-decl)

$\Els~(\calfa~p) ~\lra~ \tpi{x}{\Set}{\Els~(p~x)}$ \hfill ($\calfa$-red)
\end{leftbar}

This way, the type of the identity function is
$\tapp{\Els}{\left(\tapp{\calfa}{\left(\tabs{x}{\Set}{\tapp{\lift}{\left(x\arr
        x\right)}}\right)}\right)}$. It reduces to
$\tpi{x}{\Set}{\El~x \ra \El~x}$. Therefore, it is inhabited by the
term $\tabs{x}{\Set}{\tabs{y}{\tapp{\El}{x}}{y}}$.
In a similar way, the symbol $=$ can then be given the type
$\tapp{\Els}{\left(\tapp{\calfa}{\left(\tabs{x}{\Set}{\tapp{\lift}
      {\left(x \arr x \arr \o \right)}}\right)}\right)}$
that reduces to
$\tpi{x}{\Set}{\tapp{\El}{x} \ra \tapp{\El}{x} \ra \Prop}$.

\subsection{Prenex predicative type quantification in propositions}

When we express the reflexivity of the polymorphic equality, we need
also to quantify over a type variable, but now in a proposition.  To
be able to do so, we introduce another quantifier and its associated
rewriting rule

\begin{leftbar}
$\sffa:(\Set \ra \Prop) \ra \Prop$ \hfill ($\sffa$-decl)

$\Prf~(\sffa~p) ~\lra~ \tpi{x}{\Set}{\Prf~(p~x)}$ \hfill ($\sffa$-red)
\end{leftbar}

\noindent
This way, the reflexivity of equality can be expressed as
$\left(\tapp{\sffa}{\left(\tabs{s}{\Set}{\tapp{\tapp{\fa}{s}}{\left(\tabs{x}{\tapp{\El}{s}}{\tapp{\tapp{\tapp{=}{s}}{x}}{x}}\right)}}\right)}\right)$.

\subsection{The theory $\cU$: putting everything together}

As mentioned in Section \ref{sec-lpmc}, we call ``axiom'' a constant
declaration together with its rewrite rules if any. Hence, in the
following, we denote by ($\bliota$) the axiom consisting of
($\bliota$-decl) and ($\bliota$-red), and similarly for all the other
axioms.

The theory $\cU$ is then formed with 43 axioms: ($\I$), ($\Set$),
($\El$), ($\bliota$), ($\Prop$), ($\Prf$), ($\imp$), ($\fa$),
($\bltop$), ($\blbot$), ($\blneg$), ($\conj$), ($\disj$), ($\ex$),
($\Prfc$), ($\impc$), ($\conjc$), ($\disjc$), ($\fac$), ($\exc$),
($\o$), ($\arr$), ($\arrd$), ($\impd$), ($\blpi$), ($\0$), ($\S$),
($\pred$), ($\positive$), ($\psub$), ($\pair$), ($\paird$), ($\fst$),
($\snd$), ($\Sett$), ($\set$), ($\dtarr$), ($\Ty$), ($\Scheme$),
($\Els$), ($\lift$), ($\calfa$), ($\sffa$).

Note that, strictly speaking, the declaration ($\paird$-decl) and
the rule ($\pair$-red) are not part of the theory $\cU$, but of its
extension $\cU^{\dagger}$ used to define the conversion on the terms
of $\cU$.

Among these axioms, 14 only have a constant declaration, 27 have a
constant declaration and one rewriting rule, and 2 have a constant
declaration and two rewriting rules. So $\Sigma_{\cU}$ contains 43
declarations and $\cR_\cU$ 31 rules.

This large number of axioms is explained by the fact that \lpr{} is a
weaker framework than Predicate logic. The 20 first axioms are
needed just to construct notions that are primitive in Predicate
  logic: terms, propositions, with their 13 constructive and
classical connectives and quantifiers, and proofs.  So the theory
$\cU$ is just 23 axioms on top of the definition of Predicate
  logic.

It is also explained by the fact that axioms are more atomic than in
Predicate logic: 4 axioms are needed to express ``the'' axiom of
infinity: ($\0$), ($\S$), ($\pred$), and ($\positive$); 5 to express
predicate subtyping: ($\psub$), ($\pair$), ($\paird$), ($\fst$), and
($\snd$); 4 to express dependent types:
($\Sett$),($\set$),($\dtarr$),($\Ty$); and 5 to express prenex
polymorphism: ($\Scheme$), ($\Els$), ($\lift$), ($\calfa$), and
($\sffa$). The 5 remaining axioms express propositions as objects
($\o$); various forms of functionality: ($\arr$), ($\arrd$), and
($\blpi$); and dependent implication ($\impd$).

\section{Sub-theories}\label{sec-frag-th}

Not all proofs require all these axioms. Many proofs can be
expressed in sub-theories built by bringing together some of the
axioms of $\cU$, but not all.

Given subsets $\Sigma_\cS$ of $\Sigma_\cU$ and $\cR_\cS$ of $\cR_\cU$,
we would like to be sure that a proof in $\cU$, using only constants
in $\Sigma_\cS$, is a proof in $(\Sigma_\cS,\cR_\cS)$. Such a result is
trivial in Predicate logic: for instance, a proof in ZFC which
does not use the axiom of choice is a proof in ZF, but it is less
straightforward in \lpr, because $(\Sigma_\cS,\cR_\cS)$ might not be a
theory. So we should not consider any pair $(\Sigma_\cS,\cR_\cS)$. For
instance, as $\Set$ occurs in the type of $\El$, if we want $\El$ in
$\Sigma_\cS$, we must take $\Set$ as well.  In the same way, as
$\positive~(\S~x)$ rewrites to $\bltop$, if we want ($\positive$) and
($\S$) in $\Sigma_\cS$, we must include $\bltop$ in $\Sigma_\cS$ and the
rule rewriting $\positive~(\S~x)$ to $\bltop$ in $\cR_\cS$.

This leads to a definition of a notion of sub-theory and to prove that, if
$(\Sigma_1,\cR_1)$ is a sub-theory of a theory $(\Sigma_0,\cR_0)$,
$\Gamma$, $t$ and $A$ are in $\Lambda(\Sigma_1)$, and $\Gamma
\vdash_{\Sigma_0,\cR_0} t:A$, then $\Gamma \vdash_{\Sigma_1,\cR_1}
t:A$.

This property implies that, if $\pi$ is a proof of $A$ in $\cU$ and
both $A$ and $\pi$ are in $\Lambda(\Sigma_1)$, then $\pi$ is a proof
of $A$ in $(\Sigma_1,\cR_1)$, but it does not imply that if $A$ is in
$\Lambda(\Sigma_1)$ and $A$ has a proof in $\cU$, then it has a proof
in $(\Sigma_1,\cR_1)$.

\subsection{Fragments}

\begin{defi}[Fragment]
  A signature $\Sigma_1$ is included in a signature $\Sigma_0$,
  denoted $\Sigma_1 \subseteq \Sigma_0$, if each declaration $c:A$ of
  $\Sigma_1$ is a declaration of $\Sigma_0$.

  A system $(\Sigma_1,\cR_1)$ is a fragment of a system
  $(\Sigma_0,\cR_0)$, if the following conditions are satisfied:
  \begin{itemize}
  \item $\Sigma_1 \subseteq \Sigma_0$ and $\cR_1\subseteq\cR_0$;
  \item for all $(c:A)\in\Sigma_1$, $\const(A)\subseteq|\Sigma_1|$;
  \item for all $\ell\lra r\in\cR_0$, if $\const(\ell)\subseteq|\Sigma_1|$,
    then $\const(r)\subseteq|\Sigma_1|$ and $\ell\lra r\in\cR_1$.
  \end{itemize}
\end{defi}

We write $\vdash_i$ for $\vdash_{\Sigma_i,\cR_i}$, $\lra_i$ for
$\lra_{\beta\cR_i}$, and $\equiv_i$ for $\equiv_{\beta\cR_i}$.

\begin{lem}[Preservation of reduction]
\label{lem-red}
If $(\Sigma_1,\cR_1)$ is a fragment of $(\Sigma_0,\cR_0)$, $t \in
\Lambda(\Sigma_1)$ and $t \lra_0 u$, then $t \lra_1 u$ and $u \in
\Lambda(\Sigma_1)$.
\end{lem}

\begin{proof}
By induction on the position where the rule is applied. We only detail
the case of a top reduction, the other cases easily follow by
induction hypothesis.

So, let $\ell \lra r$ be the rule used to rewrite $t$ in $u$ and $\theta$ such
that $t = \theta \ell$ and $u = \theta r$.  As $t \in \Lambda(\Sigma_1)$,
we have $\ell \in \Lambda(\Sigma_1)$ and, for all $x$ free in $\ell$, $\theta
x \in \Lambda(\Sigma_1)$. Thus, as $(\Sigma_1,\cR_1)$ is a fragment of
$(\Sigma_0,\cR_0)$, $r\in\Lambda(\Sigma_1)$ and $\ell \lra r \in \cR_1$.
Therefore $t \lra_1 u$ and $u = \theta r \in \Lambda(\Sigma_1)$.
\end{proof}

\begin{lem}[Preservation of confluence]
\label{lem-confluence}
Every fragment of a confluent system is confluent.
\end{lem}

\begin{proof}
Let $(\Sigma_1,\cR_1)$ be a fragment of a confluent system
$(\Sigma_0,\cR_0)$.  We prove that $\lra_1$ is confluent on
$\Lambda(\Sigma_1)$.  Assume that $t, u, v \in \Lambda(\Sigma_1)$, $t
\lra_1^* u$ and $t \lra_1^* v$. Since $|\Sigma_1|\subseteq|\Sigma_0|$,
we have $t, u, v \in \Lambda(\Sigma_0)$. Since $\cR_1 \subseteq
\cR_0$, we have $t \lra_0^* u$ and $t \lra_0^* v$. By confluence of
$\lra_0$ on $\Lambda(\Sigma_0)$, there exists a $w$ in
$\Lambda(\Sigma_0)$ such that $u \lra_0^* w$ and $v \lra_0^* w$. Since
$u, v \in \Lambda(\Sigma_1)$, by Lemma \ref{lem-red}, $w \in
\Lambda(\Sigma_1)$, $u \lra_1^* w$ and $v \lra_1^* w$.
\end{proof}

\begin{defi}[Sub-theory]
A system $(\Sigma_1,\cR_1)$ is a sub-theory of a theory $(\Sigma_0,\cR_0)$,
if it is a fragment of $(\Sigma_0,\cR_0)$ and it is a theory.  
\end{defi}

As we already know that $\cR_1$ is confluent, this amounts to
say that each rule of $\cR_1$ preserves typing in $(\Sigma_1,\cR_1)$.

\subsection{The fragment theorem}

\begin{thm}\label{th-frag}
Let $(\Sigma_0,\cR_0)$ be a confluent system and $(\Sigma_1,\cR_1)$ be a
sub-theory of $(\Sigma_0,\cR_0)$.
\begin{itemize}
    \item
If the judgement $\Gamma \vdash_0 t:D$ is derivable,
$\Gamma\in\Lambda(\Sigma_1)$ and $t\in\Lambda(\Sigma_1)$, then there
exists $D'\in\Lambda(\Sigma_1)$ such that $D \lra_0^* D'$ and the
judgement $\Gamma \vdash_1 t:D'$ is derivable.

  \item
If the judgement $\vdash_0
\Gamma~\mbox{well-formed}$ is derivable and
$\Gamma\in\Lambda(\Sigma_1)$, then the judgement $\vdash_1
\Gamma~\mbox{well-formed}$ is derivable.
\end{itemize}

\end{thm}

\begin{proof}
By mutual induction on the derivations, and by case analysis on the
last typing rule. Before detailing each case, note that the most
difficult cases are (abs), (app), and (conv), the other cases are a
simple application of the induction hypothesis.

\begin{itemize}
\item If the last rule of the derivation is
$$\irule{\Gamma\vdash_0 A:\TYPE & \Gamma,x:A \vdash_0 B:s &
         \Gamma,x:A\vdash_0 t:B}
        {\Gamma\vdash_0 \tabs{x}{A}{t} : \tpi{x}{A}{B}}
        {\mbox{(abs)}}$$
as $\Gamma$, $A$, and $t$ are in $\Lambda(\Sigma_1)$, by induction
hypothesis, there exists $A'$ in $\Lambda(\Sigma_1)$ such that $\TYPE
\lra^*_0 A'$ and $\Gamma \vdash_1 A:A'$ is derivable, and there exists
$B'$ in $\Lambda(\Sigma_1)$ such that $B \lra_0^* B'$ and $\Gamma,x:A
\vdash_1 t:B'$ is derivable.  As $\TYPE$ is a sort, $A' =
\TYPE$. Therefore, $\Gamma \vdash_1 A:\TYPE$ is derivable.

As $B$ is typable and every subterm of a typable term is typable,
$\KIND$ does not occur in $B$. As $B \lra_0^* B'$ and no rule contains
$\KIND$, $\KIND$ does not occur in $B'$ as well. Hence, $B'\neq
\KIND$. By Lemma \ref{typtypprod}, as $\Gamma,x:A \vdash_1 t:B'$ is
derivable and $B'\neq \KIND$, there exists a sort $s'$ such that
$\Gamma,x:A \vdash_1 B':s'$ is derivable.

Thus, by the rule (abs), $\Gamma\vdash_1\tabs{x}{A}{t}:\tpi{x}{A}{B'}$
is derivable.  So there is $D'= \tpi{x}{A}{B'}$ in
$\Lambda(\Sigma_1)$ such that $\tpi{x}{A}{B} \lra_0^* D'$ and
$\Gamma\vdash_1\tabs{x}{A}{t}:D'$ is derivable.

\item If the last rule of the derivation is
$$\irule{\Gamma\vdash_0t:\tpi{x}{A}{B} & \Gamma\vdash_0u:A}
        {\Gamma\vdash_0 t~u:(u/x)B}
        {\mbox{(app)}}$$
as $\Gamma$, $t$, and $u$ are in $\Lambda(\Sigma_1)$, by induction
hypothesis, there exist $C$ and $A_2$ in $\Lambda(\Sigma_1)$, such
that $\tpi{x}{A}{B} \lra_0^* C$, $\Gamma \vdash_1 t:C$ is derivable,
$A \lra_0^*A_2$, and $\Gamma \vdash_1u :A_2$ is derivable.  As
$\tpi{x}{A}{B} \lra_0^* C$ and rewriting rules are of the form
$(c~l_1\ldots l_n\lra r)$, there exist $A_1$ and $B_1$ in
$\Lambda(\Sigma_1)$ such that $C = \tpi{x}{A_1}{B_1}$, $A \lra_0^*
A_1$, and $B \lra_0^* B_1$. By confluence of $\lra_0$, there exists
$A'$ such that $A_1 \lra_0^* A'$ and $A_2 \lra_0^*A'$. By Lemma
\ref{lem-red}, as $A_1 \in\Lambda(\Sigma_1)$ and $A_1 \lra_0^* A'$, we
have $A' \in \Lambda(\Sigma_1)$ and $A_1 \lra_1^* A'$.  In a similar
way, as $A_2 \in \Lambda(\Sigma_1)$ and $A_2 \lra_0^* A'$, we have
$A_2 \lra_1^* A'$.  By Lemma \ref{typtypprod}, as $\Gamma \vdash_1
t:\tpi{x}{A_1}{B_1}$ is derivable and $\tpi{x}{A_1}{B_1} \neq \KIND$,
there exists a sort $s$ such that $\Gamma \vdash_1
\tpi{x}{A_1}{B_1}:s$ is derivable.  Thus, by Lemma \ref{typtypprod},
$\Gamma \vdash_1 A_1:\TYPE$ is derivable.

As $\Gamma \vdash_1 \tpi{x}{A_1}{B_1}:s$, $\tpi{x}{A_1}{B_1} \lra_1^*
\tpi{x}{A'}{B_1}$, and $(\Sigma_1,\cR_1)$ preserves typing, $\Gamma
\vdash_1 \tpi{x}{A'}{B_1}:s$ is derivable.  In a similar way, as
$\Gamma \vdash_1 A_1:\TYPE$ is derivable, and $A_1 \lra_1^* A'$,
$\Gamma \vdash_1 A':\TYPE$ is derivable. Therefore, by the rule
(conv), $\Gamma\vdash_1 t:\tpi{x}{A'}{B_1}$ and $\Gamma\vdash_1u:A'$
are derivable. Therefore, by the rule (app), $\Gamma \vdash_1
t~u:(u/x)B_1$ is derivable. So there exists $D' = (u/x)B_1$ in
$\Lambda(\Sigma_1)$, such that $(u/x)B \lra_0^* D'$ and
$\Gamma\vdash_1t~u:D'$ is derivable.

\item If the last rule of the derivation is
$$\irule{\Gamma \vdash_0 t:A & \Gamma \vdash_0 B:s}
        {\Gamma \vdash_0 t:B}
        {\mbox{(conv)}~~~A \equiv_{\beta\cR_0} B}$$
as $\Gamma$ and $t$ are in $\Lambda(\Sigma_1)$, by induction
hypothesis, there exists $A'$ in $\Lambda(\Sigma_1)$ such that $A
\lra_0^* A'$ and $\Gamma \vdash_1 t:A'$ is derivable.  By confluence
of $\lra_0$, there exists $C$ such that $A' \lra_0^* C$ and $B
\lra_0^* C$.  As $A' \in \Lambda(\Sigma_1)$ and $A' \lra_0^* C$ we
have, by Lemma \ref{lem-red}, $C \in \Lambda(\Sigma_1)$ and $A'
\lra_1^* C$.

As $B$ is typable and every subterm of a typable term is typable,
$\KIND$ does not occur in $B$. As $B\lra_0^*C$ and no rule contains
$\KIND$, $\KIND$ does not occur in $C$ as well. Thus $C \neq \KIND$.
As $A' \lra_0^* C$, $A' \neq \KIND$.  By Lemma
\ref{typtypprod}, as $\Gamma \vdash_1 t:A'$ and $A' \neq \KIND$, there
exists a sort $s'$ such that $\Gamma\vdash_1 A':s'$ is derivable.
Thus, as $A' \lra_1^* C$,
and $(\Sigma_1,\cR_1)$ preserves typing, $\Gamma \vdash_1 C:s'$ is derivable.
a
As $\Gamma \vdash_1 t:A'$ and $\Gamma \vdash_1 C:s'$ are derivable and
$A'\lra_1C$, by the rule (conv), $\Gamma \vdash_1 t:C$ is derivable.
Thus there exists $D' = C$ in $\Lambda(\Sigma_1)$ such that $\Gamma
\vdash_1 t:D'$ is derivable and $B \lra_0^* D'$.

  \item If the last rule of the derivation is
$$\irule{}
        {\vdash_0 [~]~\mbox{well-formed}}
        {\mbox{(empty)}}$$
by the rule
(empty), $\vdash_1
      [~]~\mbox{well-formed}$ is derivable.

\item If the last rule of the derivation is
$$\irule{\Gamma \vdash_0 A:s}
        {\vdash_0 \Gamma, x:A~\mbox{well-formed}}
        {\mbox{(decl)}}$$
as $\Gamma$ and $A$ are in $\Lambda(\Sigma_1)$, by induction
hypothesis, there exist $A'$ in $\Lambda(\Sigma_1)$ such that $s
\lra_0^* A'$ and $\Gamma \vdash_1 A:A'$ is
derivable. As $s$ is a sort, $A' = s$. Therefore,
$\Gamma \vdash_1 A:s$ is derivable and, by the rule (decl),
$\vdash_1 \Gamma, x:A~\mbox{well-formed}$ is derivable.

\item If the last rule of the derivation is
$$\irule{\vdash_0 \Gamma~\mbox{well-formed}}
        {\Gamma \vdash_0 \TYPE:\KIND}
        {\mbox{(sort)}}$$
as $\Gamma$ is in $\Lambda(\Sigma_1)$, by induction
hypothesis, $\vdash_1 \Gamma~\mbox{well-formed}$
is derivable. Thus, by the rule (sort),
$\Gamma \vdash_1 \TYPE:\KIND$ is derivable. So there exists
$D' = \KIND$ in $\Lambda(\Sigma_1)$
such that $\KIND \lra_0^* D'$ and $\Gamma \vdash_1 \TYPE:D'$.

\item If the last rule of the derivation is
$$\irule{\vdash_0 \Gamma~\mbox{well-formed}~~~\vdash_0 A:s}
        {\Gamma \vdash_0 c:A}
        {\mbox{(const)}~~~c:A \in \Sigma_0}$$
as $\Gamma$ is in $\Lambda(\Sigma_1)$, by induction hypothesis,
$\vdash_1 \Gamma~\mbox{well-formed}$ is
derivable.
And as $c$ is in $\Lambda(\Sigma_1)$, it is in
$|\Sigma_1|$, thus $c:A$ is in $\Sigma_1$ and, since $(\Sigma_1,\cR_1)$
is a fragment of $(\Sigma_0,\cR_0)$, $A \in \Lambda(\Sigma_1)$.

Thus, by induction hypothesis, there exists $A'$ such that
$\vdash_1 A:A'$ is derivable and $s \lra_0^* A'$. As $s$ is a sort, $A' = s$.
So $\vdash_1 A:s$ is derivable.
Thus, by the rule (const),
$\Gamma \vdash_1 c:A$ is derivable. So, there exists $D' = A$ in $\Lambda(\Sigma_1)$ such that $A \lra_0^* D'$ and $\Gamma \vdash_1 c:D'$ is derivable.

\item If the last rule of the derivation is
$$\irule{\vdash_0 \Gamma~\mbox{well-formed}}
        {\Gamma \vdash_0 x:A}
        {\mbox{(var)}~~~x:A \in \Gamma}$$
as $\Gamma$ is in $\Lambda(\Sigma_1)$, by induction hypothesis, $\vdash_1 \Gamma~\mbox{well-formed}$ is
derivable. Thus, by the rule (var), $\Gamma \vdash_1
x:A$ is derivable.  So there exists $D' = A$ in $\Lambda(\Sigma_1)$
such that $A \lra_0^* D'$ and $\Gamma \vdash_1 x:D'$.

\item If the last rule of the derivation is
  $$\irule{\Gamma \vdash_0 A:\TYPE & \Gamma, x:A\vdash_0 B:s}
        {\Gamma \vdash_0 \tpi{x}{A}{B}:s}
        {\mbox{(prod)}}$$
as $\Gamma$, $A$, and $B$ are in $\Lambda(\Sigma_1)$, by induction
hypothesis, there exists $A'$ in $\Lambda(\Sigma_1)$ such that $\TYPE
\lra_0^* A'$ and $\Gamma \vdash_1 A:A'$ is derivable and
there exists $B'$ in $\Lambda(\Sigma_1)$ such that $s \lra_0^*B'$ and
$\Gamma,x:A\vdash_1 B:B'$ is derivable.  As $\TYPE$
and $s$ are sorts, $A' = \TYPE$ and $B' = s$. Therefore,
$\Gamma \vdash_1 A:\TYPE$ and $\Gamma,x:A \vdash_1 B:s$ are
derivable. Thus, by the rule (prod), $\Gamma\vdash_1
\tpi{x}{A}{B}:s$ is derivable. So there exists $D' = s$ in
$\Lambda(\Sigma_1)$ such that $s \lra_0^* D'$ and
$\Gamma\vdash_1 \tpi{x}{A}{B}:D'$ is derivable. \qedhere
\end{itemize}
\end{proof}

\begin{cor}
\label{cor-subtheory}
Let $(\Sigma_0,\cR_0)$ be a confluent system, $(\Sigma_1,\cR_1)$ be a
sub-theory of $(\Sigma_0,\cR_0)$.  If
$\Gamma\vdash_0t:D$, $\Gamma\in\Lambda(\Sigma_1)$,
$t\in\Lambda(\Sigma_1)$, and $D \in \Lambda(\Sigma_1)$, then
$\Gamma\vdash_1t:D$.

In particular, if $(\Sigma_0,\cR_0)$ is a theory, $(\Sigma_1,\cR_1)$ is a
sub-theory of $(\Sigma_0,\cR_0)$, $\Gamma\vdash_0t:D$,
$\Gamma\in\Lambda(\Sigma_1)$, $t\in\Lambda(\Sigma_1)$, and $D \in
\Lambda(\Sigma_1)$, then $\Gamma\vdash_1t:D$.
\end{cor}

\begin{proof}
There is a $D'\in\Lambda(\Sigma_1)$ such that $D\lra_0^*D'$ and
$\Gamma\vdash_1t:D'$.  As $D\in \Lambda(\Sigma_1)$ and $D\lra_0^*D'$.
By Lemma \ref{lem-red} we have $D\lra_1^*D'$, and we conclude with the
rule (conv).
\end{proof}

\begin{thm}[Sub-theories of $\cU$]\label{lem-frag-U}
Every fragment $(\Sigma_1,\cR_1)$ of $\cU$ (including $\cU$ itself) is a
theory, that is, is confluent and preserves typing.
\end{thm}

\begin{proof}
The relation $\lra_{\beta\cR_\cU}$ is confluent on
$\Lambda(\Sigma_\cU)$ since it is an orthogonal combinatory reduction
system \cite{klop93tcs}. Hence, after the fragment theorem, it is
sufficient to prove that every rule of $\cR_\cU$ preserves typing in
any fragment $(\Sigma_1,\cR_1)$ containing the symbols of the rule.

To this end, we will use the criterion described in \cite[Theorem
  19]{blanqui20fscd} which consists in computing the equations that
must be satisfied for a rule left-hand side to be typable, which are
system-independent, and then check that the right-hand side has the
same type modulo these equations in the desired system: for all rules
$l\lra r\in\Lambda(\Sigma_1)$, sets of equations $\cE$ and terms $T$,
if the inferred type of $l$ is $T$, the typability constraints of $l$
are $\cE$, and $r$ has type type $T$ in the system $\Lambda(\Sigma_1)$
whose conversion relation $\equiv_{\beta\cR\cE}$ has been enriched
with $\cE$, then $l\lra r$ preserves typing in $\Lambda(\Sigma_1)$.

This criterion can easily be checked for all the rules but
($\pred$-red2) and ($\fst$-red) because, except in those two cases,
the left-hand side and the right-hand side have the same type.

In ($\pred$-red2), $\pred~(\S~x) \lra x$, the left-hand side has type
$\I$ if the equation $type(x)=\I$ is
satisfied. Modulo this equation, the right-hand side has type
$\I$ in any fragment containing the symbols of the rule.

In ($\fst$-red), $\fst~t~p~(\paird~t'~p'~m) \lra m$, the left-hand
side has type $\El~t$ if $type(t) = Set$, $type(p) = \El~t \ra Prop$,
$\El~(\psub~t'~p') = \El~(\psub~t~p)$, $type(t') = \Set$, $type(p') =
\El~t' \ra Prop$, and $type(m) = \El~t'$. But, in $\cU$, there is no
rule of the form $\El~(\psub~t~p) \lra r$. Hence, by confluence, the
equation $\El~(\psub~t'~p') = \El~(\psub~t~p)$ is equivalent to the
equations $t' = t$ and $p' = p$. Therefore, the right-hand side is of
type $\El~t$ in every fragment of $\cU$ containing the symbols of the
rule.
\end{proof}

\section{Examples of sub-theories of the theory $\cU$}\label{sec-subth}
\label{sec:examples_sub-theories}

We finally identify 15 sub-theories of the theory $\cU$, that
correspond to known theories. For each of these sub-theories
$(\Sigma_\cS,\cR_\cS)$, according to the Corollary
\ref{cor-subtheory}, if $\Gamma$, $t$, and $A$ are in
$\Lambda(\Sigma_\cS)$, and $\Gamma \vdash_{\Sigma_\cU,\cR_\cU} t
: A$, then $\Gamma \vdash_{\cR_\cS,\Sigma_\cS} t : A$.

\begin{figure}[ht]
  \centering
  \def\svgwidth{0.8\textwidth}
  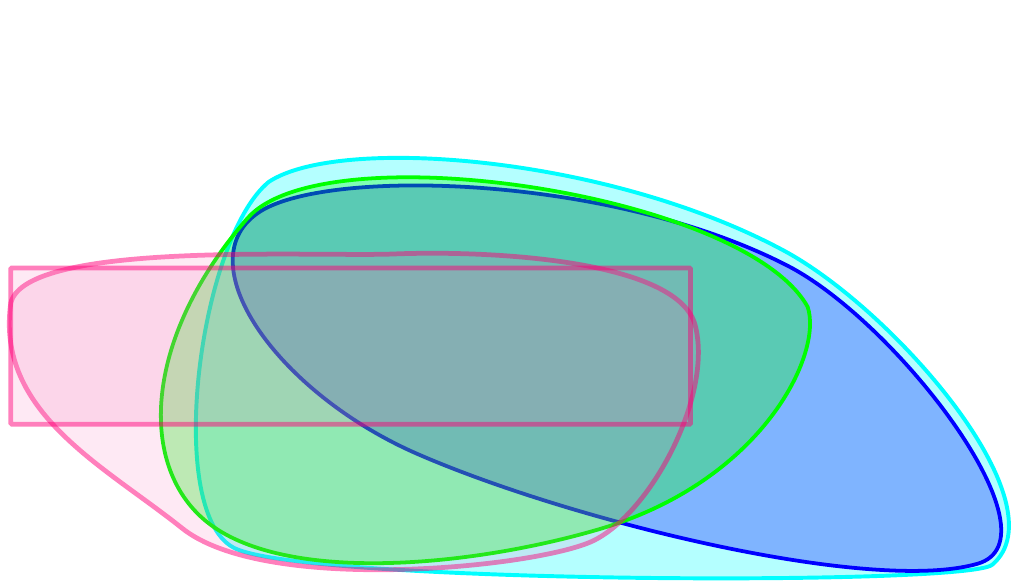
  \caption{{\textbf{The wind rose.}}
    In black: Minimal, Constructive, and Ecumenical predicate logic.
    In orange: Minimal, Constructive, and Ecumenical simple type theory.
    In green: Simple type theory with prenex polymorphism.
    In blue: Simple type theory with predicate subtyping.
    In cyan: Simple type theory with predicate subtyping
    and prenex polymorphism.
    In pink: the Calculus of constructions with
    a constant $\bliota$, without and with prenex polymorphism.}
\end{figure}

\subsection{Minimal predicate logic}
\label{sec:predicate-logic:min}

The 8 axioms ($\I$), ($\Set$), ($\El$), ($\bliota$), ($\Prop$), ($\Prf$),
($\imp$), and ($\fa$) define Minimal predicate logic
$(\Sigma_\cM,\cR_\cM)$.
$$\begin{array}{rcl}
\I &:& \TYPE\\
\Set &:& \TYPE\\
\El &:& \Set \ra \TYPE\\
\bliota &:& \Set\\
\El~\bliota &\lra& \I\\
\Prop &:& \TYPE\\
\Prf &:& \Prop \ra \TYPE\\
\imp &:& \Prop \ra \Prop \ra \Prop\\
\Prf~(x \imp y) &\lra& \Prf~x \ra \Prf~y\\
\fa&:&\tpi{x}{\Set}{(\El~x \ra \Prop) \ra \Prop}\\
\Prf~(\fa~x~p) &\lra& \tpi{z}{\El~x}{\Prf~(p~z)}\\
\end{array}$$
We could save the declaration ($\I$-decl) and the rule ($\bliota$-red)
by using $\El~\bliota$ instead of $\I$.

This theory can be proven equivalent to more common formulations of
Minimal predicate logic.  To do so, consider a language $\cL$ in
predicate logic.  We define a corresponding \lpr{} context
$\Gamma_\cL$ containing for each constant $f$ of $\cL$ a constant $f$
of type $\I \ra \dots \ra \I \ra \I$ and for each predicate symbol $P$
of $\cL$ a constant $P$ of type $\I \ra \dots \ra \I \ra \Prop$. A
term (resp. a proposition) of minimal predicate logic $t$ of $\cL$
translates in the natural way to a \lpr{} term of type $\I$
(resp. $\Prop$) in the theory $(\Sigma_\cM,\cR_\cM)$ and in the context
$\Gamma_\cL, \Delta$, where $\Delta$ contains, for each variable $x$
free in $t$, a variable $x$ of type $\I$.  We use the same notation
for the term (resp. the proposition) and its translation.

\begin{thm}
\label{th:predicate-logic:min:soundness_cons}
Let $\cL$ be a language and $A_1, ..., A_n \vdash B$ be a sequent of
minimal predicate logic in $\cL$.  Let $\Gamma_\cL$ containing for
each constant $f$ of $\cL$ a constant $f$ of type $\I \ra \dots \ra \I
\ra \I$ and for each predicate symbol $P$ of $\cL$ a constant $P$ of
type $\I \ra \dots \ra \I \ra \Prop$.  Let $\Delta$ be a context
containing for each variable $x$ free in $A_1, ..., A_n \vdash B$, a
variable $x$ of type $\I$.  Let $\Delta'$ be a context containing, for
each hypothesis $A_i$, a variable $a_i$ of type $\Prf~A_i$.

Then, the sequent
$A_1, ..., A_n \vdash B$ has a proof in minimal logic, if and only
if there exists a \lpr{} term $\pi$ such that $\Gamma_\cL, \Delta,
\Delta' \vdash_{\Sigma_\cM,\cR_\cM} \pi : \Prf~B$.
\end{thm}

\begin{proof}
The left-to-right implication is a trivial induction on the structure of the proof.

For the converse, it is enough to consider an irreducible term $\pi$ of type
$\Prf~B$ since one can prove that $\lra_{\beta\cR_\cM}$ terminates, by applying \cite{blanqui19fscd} for instance. We then prove, by induction on $\pi$, that the sequent $A_1,
..., A_n \vdash B$ has a proof in minimal logic.  As $\pi$ has the
type $\Prf~B$, it is neither a sort, nor a product, thus it is either
an abstraction or a term of the form $z~\rho_1~...~\rho_p$.

\begin{itemize}
\item If $\pi$ is an abstraction then $\Prf~B$ is equivalent to a
  product. Hence, $B$ either has the form $C \imp D$ or
  $\fa~\iota~\tabs{x}{\I}{D}$. In the first case $\pi =
  \tabs{x}{\Prf~C}{\pi'}$ and $\pi'$ is a term of type $\Prf~D$ in
  $\Gamma_\cL, \Delta, \Delta', x:\Prf~C$.  By induction hypothesis,
  the sequent $A_1, ..., A_n, C \vdash D$ has a proof and so does the
  sequent $A_1, ..., A_n \vdash C \imp D$.  In the second $\pi =
  \tabs{x}{\I}{\pi'}$ and $\pi'$ is a term of type $\Prf~D$
  in $\Gamma_\cL, \Delta, x:\I, \Delta'$. By induction
  hypothesis, the sequent $A_1, ..., A_n \vdash D$ has a proof and so
  does the sequent $A_1, ..., A_n \vdash \fa~\iota~
  \tabs{x}{\I}{D}$.

\item If $\pi$ has the form $z~\rho_1~...~\rho_p$, then as it has
  the type $\Prf~B$, $z$ can neither be a constant
  of $\Sigma_M$, nor a variable of $\Gamma_\cL$, $\Delta$. Hence,
  it is a variable
  of $\Delta'$.  Thus, it has the type $\Prf~A_i$
  for some $i$.  We prove, by induction on $j$ that the term
  $z~\rho_1~...~\rho_j$ has the type $\Prf~C$ for some
  proposition $C$, such that the sequent $A_1, ..., A_n \vdash C$ has
  a proof.  For $j = 0$, the sequent $A_1, ..., A_n \vdash A_i$ has a
  proof.  Assume the property holds for $j$. Then, as the term
 $z~\rho_1~...~\rho_j~\rho_{j+1}$ is well typed, the type
  $\Prf~C$ is a product type and $C$ either has the form $D \imp E$
  or $\fa~\iota~\tabs{x}{\I}{E}$. In the first case
  $\rho_{j+1}$ is a term of type $\Prf~D$, by induction
  hypothesis, the sequent $A_1, ..., A_n \vdash D$ has a proof, hence
  the term $z~\rho_1~...~\rho_j~\rho_{j+1}$ has the type
  $\Prf~E$ and the sequent $A_1, ..., A_n \vdash E$ has a
  proof.  In the second case, $\rho_{j+1}$ is an irreducible term of type
  $\I$, thus it is an object-term, the term
  $z~\rho_1~...~\rho_j~\rho_{j+1}$ has the type
  $\Prf~(\rho_{j+1}/x)E$, and the sequent $A_1, ..., A_n
  \vdash (\rho_{j+1}/x)E$ has a proof. \qedhere
\end{itemize}
\end{proof}

As Minimal predicate logic is itself a
logical framework, it must be complemented with more axioms,
such as the axioms of geometry, arithmetic, etc.

\subsection{Constructive predicate logic}
\label{sec:predicate-logic:constructive}

The 14 axioms ($\I$),($\Set$), ($\El$), ($\bliota$), ($\Prop$), ($\Prf$),
($\imp$), ($\fa$), ($\bltop$), ($\blbot$), ($\blneg$), ($\conj$),
($\disj$), and ($\ex$) define Constructive predicate logic.  This
theory can be proven equivalent to more common formulations of
Constructive predicate logic \cite{Dorra,expressing}.

\subsection{Ecumenical predicate logic}
\label{sec:predicate-logic:ecumenical}

The 20 axioms ($\I$), ($\Set$), ($\El$), ($\bliota$), ($\Prop$),
($\Prf$), ($\imp$), ($\fa$), ($\bltop$), ($\blbot$), ($\blneg$),
($\conj$), ($\disj$), ($\ex$), ($\Prfc$), ($\impc$), ($\conjc$),
($\disjc$), ($\fac$), and ($\exc$) define Ecumenical predicate logic.
This theory can be proven equivalent to more common formulations of
Ecumenical predicate logic \cite{Grienenberger21}.

Note that classical predicate logic is not a sub-theory of the theory
$\cU$, because the classical connectives and quantifiers depend on the
constructive ones. Yet, it is known that if a proposition contains
only classical connectives and quantifiers, it is provable in
Ecumenical predicate logic if and only if it is provable in classical
predicate logic.

\subsection{Minimal simple type theory}
\label{sec:stt:min}

Adding the two axioms ($\o$) and ($\arr$) to Predicate logic defines
Simple type theory. Indeed, Simple type theory is the theory of
propositional contents and functions. A simple type $T$ is naturally
translated to \lpr{} as a term $T$ of type $\Set$, using types
$\bliota$ and $\o$ and the arrow construction $\arr$. The higher order
terms are shallowly translated: $\lambda$-abstractions and
applications are translated using respectively \lpr{}'s
$\lambda$-abstractions and applications.

The 10 axioms ($\I$), ($\Set$), ($\bliota$), ($\El$), ($\Prop$),
$(\Prf)$, ($\imp$), ($\fa$), ($\o$), and ($\arr$) define Minimal
simple type theory. And this theory can be proven equivalent to more
common formulations of Minimal simple type theory
\cite{assaf15phd,expressing}.  Just like we can save the declaration
($\I$-decl) and the rule ($\bliota$-red) by replacing everywhere $\I$
with $\El~\bliota$, we could save the declaration ($\Prop$-decl) and
the rule ($\o$-red) by replacing everywhere $\Prop$ with $\El~\o$.
This presentation is the usual presentation of Simple type theory in
\lpr{} \cite{expressing} with 8 declarations and 3 reduction
rules. However, doing so, the obtained presentation of
Simple type theory is not an extension of Minimal predicate logic anymore.

\subsection{Constructive simple type theory}

The 16 axioms ($\I$), ($\Set$), ($\El$), ($\bliota$), ($\Prop$),
($\Prf$), ($\imp$), ($\fa$), ($\bltop$), ($\blbot$), ($\blneg$),
($\conj$), ($\disj$), ($\ex$), ($\o$) and ($\arr$) define Constructive
simple type theory.

\subsection{Ecumenical simple type theory}

The 22 axioms ($\I$), ($\Set$), ($\El$), ($\bliota$), ($\Prop$), ($\Prf$),
($\imp$), ($\fa$), ($\bltop$), ($\blbot$), ($\blneg$), ($\conj$),
($\disj$), ($\ex$), ($\Prfc$), ($\impc$), ($\conjc$), ($\disjc$),
($\fac$), ($\exc$), ($\o$) and ($\arr$) define Ecumenical simple type
theory. And this theory can be proven equivalent to more common
formulations of Ecumenical simple type theory \cite{Grienenberger21}.

\subsection{Simple type theory with predicate subtyping}

Adding to the 10 axioms of Minimal simple type theory
the 5 axioms of
predicate subtyping yields Minimal simple type theory with predicate
subtyping, formed with the 15 axioms $(\I)$, ($\Set$), ($\bliota$), ($\El$),
($\Prop$), $(\Prf)$, ($\imp$), ($\fa$), ($\o$), ($\arr$), ($\psub$),
($\pair$), ($\paird$), ($\fst$), and ($\snd$).  This theory has been
studied by F.~Gilbert \cite{Gilbert} as a verifiable language for a
minimal and idealised version of the type system of the PVS proof
assistant. It can be proven equivalent to more common formulations of
Minimal simple type theory with predicate subtyping
\cite{Gilbert,blanqui20types}.  Such formulations like PVS
\cite{pvs:semantics} often use predicate subtyping implicitly to
provide a lighter syntax without ($\pair$), ($\paird$), ($\fst$)
nor ($\snd$) but at the expense of losing uniqueness of type and
making type-checking undecidable. In these cases, terms generally do
not hold the proofs needed to be of a sub-type, which provides proof
irrelevance.  Our implementation of proof irrelevance
of Section~\ref{sec:proof-irrelevance}
extends the conversion in order to ignore these proofs.

\subsection{Simple type theory with prenex predicative polymorphism}

Adding to the 10 axioms of Minimal simple type theory
the 5 axioms of prenex
predicative polymorphism yields Simple type theory with prenex
predicative polymorphism (STT$\forall$)
\cite{Thire-lfmtp,thire20phd} formed with the 15 axioms ($\I$), ($\Set$),
($\El$), ($\bliota$), ($\Prop$), $(\Prf)$, ($\imp$), ($\fa$), ($\o$),
($\arr$), ($\Scheme$), ($\Els$), ($\lift$), ($\calfa$), and ($\sffa$).

\subsection{Simple type theory with predicate subtyping and prenex polymorphism}

Adding to the 10 axioms of Minimal simple type theory
both the 5 axioms of
predicate subtyping and the 5 axioms of prenex polymorphism yields a
sub-theory with 20 axioms which is a subsystem of PVS~\cite{pvs:semantics}
handling both predicate subtyping and prenex polymorphism.

\subsection{The Calculus of constructions}

Pure type systems \cite{Berardi,Terlouw,barendregt92chapter} are a
family of typed $\lambda$-calculi. An example is the $\lambda
\Pi$-calculus, the $\lambda$-calculus with dependent types, which is
at the basis of \lpr{} itself. As we have seen, in \lpr{}, we
have two constants, $\TYPE$ and $\KIND$, $\TYPE$ has type $\KIND$, and
we can build a product type $\tpi{x}{A}{B}$ when both $A$ and
$B$ have type $\TYPE$, in which case the product type $\tpi{x}{A}{B}$
itself has type $\TYPE$ or when $A$ has type $\TYPE$ and $B$ has type
$\KIND$, in which case the result has type $\KIND$.

A Pure type system, in general, is defined with a set of symbols, such
as $\TYPE$ and $\KIND$, called ``sorts'', a set of axioms of
the form $\langle s_1, s_2 \rangle$, expressing that the sort $s_1$
has type $s_2$, for example $\langle \TYPE, \KIND \rangle$, and a set
of rules of the form $\langle s_1, s_2, s_3 \rangle$, expressing
that we can build the product type $\tpi{x}{A}{B}$, when $A$ has type $s_1$
and $B$ has type $s_2$, and that the product type itself has type $s_3$,
for example $\langle \TYPE, \TYPE, \TYPE \rangle$ and $\langle \TYPE,
\KIND, \KIND \rangle$. When the set of axioms is
functional, each sort has at most one type and when the set of rules
is functional, each product type has at most one type.  In this case the
Pure type system is said to be ``functional''.

To have more compact notation, we often write $*$ for the sort $\TYPE$
and $\Box$ for the sort $\KIND$. So the $\lambda \Pi$-calculus is
defined with the sorts $*$ and $\Box$, the axiom $\langle *, \Box
\rangle$, and the rules $\langle *, *, * \rangle$ and $\langle *,
\Box, \Box \rangle$.  Adding the rules $\langle \Box, *, * \rangle$
and $\langle \Box, \Box, \Box \rangle$ yields the Calculus of
constructions \cite{CoquandHuet88}.

All functional Pure type systems can be expressed in \lpr{}
\cite{cousineau07tlca}: for each sorts $s$, we introduce two constants $U_s$ of
type $\TYPE$ and $\varepsilon_s$ of type $U_s \ra \TYPE$, for each
axiom $\langle s_1, s_2 \rangle$, a constant $\dot{s}_1$ of type
$U_{s_2}$, and a reduction rule
$$\varepsilon_{s_2}~\dot{s}_1 \lra U_{s_1}$$
and
for each rule $\langle s_1, s_2, s_3 \rangle$, a constant
$\dot{\Pi}_{\langle s_1,s_2,s_3 \rangle}$ of type $\tpi{x}{U_{s_1}}
{(\varepsilon_{s_1}~x \ra U_{s_2}) \ra U_{s_3}}$
and a reduction rule
$$\varepsilon_{s_3}~(\dot{\Pi}_{\langle s_1,s_2,s_3 \rangle}~x~y) \lra
\tpi{z}{\varepsilon_{s_1}~x}{\varepsilon_{s_2}~(y~z)}$$ We obtain this
way a correct and conservative expression of the Pure type system
\cite{cousineau07tlca,expressing}.  For instance, the expression of
the Calculus of constructions yields 9 declarations and 5 rules.
Writing $\Prop$ for $U_{*}$, $\Prf$ for $\varepsilon_{*}$, $\Set$ for
$U_{\Box}$, $\El$ for $\varepsilon_{\Box}$, $\o$ for $\dot{*}$,
$\impd$ for $\dot{\Pi}_{\langle *,*,* \rangle}$, $\fa$ for
$\dot{\Pi}_{\langle \Box,*,* \rangle}$, $\blpi$ for
$\dot{\Pi}_{\langle *,\Box,\Box \rangle}$, and $\arrd$ for
$\dot{\Pi}_{\langle \Box,\Box,\Box \rangle}$ we get exactly the 9
axioms ($\Prop$), $(\Prf)$, ($\Set$), ($\El$), ($\o$), ($\impd$),
($\blpi$), ($\fa$), and ($\arrd$) and this theory is equivalent to
more common formulations of the Calculus of constructions.

As $\impd$ is $\dot{\Pi}_{\langle *,*,* \rangle}$, $\fa$ is
$\dot{\Pi}_{\langle \Box,*,* \rangle}$, $\blpi$ is $\dot{\Pi}_{\langle
  *,\Box,\Box \rangle}$, and $\arrd$ is $\dot{\Pi}_{\langle
  \Box,\Box,\Box \rangle}$, using the terminology of Barendregt's
$\lambda$-cube \cite{barendregt92chapter}, the axiom ($\blpi$) expresses
dependent types, the axiom ($\fa$) polymorphism, and the axiom
($\arrd$) type constructors.

Note that these constants have similar types
$$\begin{array}{rcl}
\impd&:& \tpi{x}{\Prop}{(\Prf~x \ra \Prop) \ra \Prop}\\
\blpi&:&\tpi{x}{\Prop}{(\Prf~x \ra \Set) \ra \Set}\\
\fa&:&\tpi{x}{\Set}{(\El~x \ra \Prop) \ra \Prop}\\
\arrd&:&\tpi{x}{\Set}{(\El~x \ra \Set) \ra \Set}\\
\end{array}$$

So if $\Gamma$ is a context and $A$ is a term in the Calculus of
constructions then $A$ is inhabited in $\Gamma$ in the Calculus of
constructions if and only if the translation $\|A\|$ of $A$ in \lpr{}
is inhabited in the translation $\|\Gamma\|$ of $\Gamma$ in \lpr{}
\cite{cousineau07tlca,expressing}.  So, the formulation of the
Calculus of constructions in \lpr{} is a conservative extension of the
original formulation of the Calculus of constructions.  In the context
$\|\Gamma\|$, variables have a \lpr{} type of the form $\Prf~u$ or
$\El~u$, and none of them can have the type $\Set$. However, in
\lpr{}, nothing prevents from declaring a variable of type
$\Set$. Hence, in \lpr{}, the judgement $x:\Set \vdash x:\Set$ can be
derived, but it is not in the image of the encoding.

\subsection{The Calculus of constructions with variables of type $\Box$}

To allow the declaration of variables of type $\Box$ in the Calculus
of constructions, a possibility is to add a sort $\triangle$ and an
axiom $\Box:\triangle$ \cite{Geuvers95}, making the sort $\triangle$ a
singleton sort that contains only one closed irreducible term: $\Box$.

Expressing this Pure type system in \lpr{} introduces
two declarations for the sort $\triangle$
and one declaration and one rule for the axiom $\Box:\triangle$
\begin{itemize}
\item $U_{\triangle}:\TYPE$
\item $\varepsilon_{\triangle} : U_{\triangle} \ra \TYPE$
\item $\dot{\Box} : U_{\triangle}$
\item $\varepsilon_{\triangle}~\dot{\Box} \lra \Set$
\end{itemize}
and the variables in a context $\|\Gamma\|$ now have the type
$\Prf~u$, $\El~u$, or $\varepsilon_{\triangle}~u$.
Just like $U_{*}$ is written $\Prop$ in $\cU$, $U_{\triangle}$
is written $\Sett$, $\varepsilon_{\triangle}$ is written $\Ty$,
and $\dot{\Box}$ is written $\set$.

But this theory can be simplified. Indeed, just like $\Box$ is the
only closed irreducible term of type $\triangle$, $\set$ is the
only closed irreducible term of type $\Sett$ and thus for any
closed term $t$ of type $\Sett$, the term
$\Ty~ t$ reduces to $\Set$.  So, we can replace
everywhere the terms of the form $\Ty~ t$ with
$\Set$ and drop the symbols $\Ty$ and $\set$
and the rule $\Ty~ \set \lra \Set$. Then, we
can drop the symbol $\Sett$ as well.

So the only difference with the Calculus of constructions without
$\triangle$ is that translations of contexts now contain variables of
type $\Set$, that translate the variables of type $\Box$.

\subsection{The Calculus of constructions with a constant $\bliota:\Box$}

Adding the axiom ($\bliota$) to the Calculus of constructions yields a
sub-theory with the 10 axioms ($\Set$), ($\El$), ($\bliota$),
($\Prop$), $(\Prf)$, ($\impd$), ($\fa$), ($\o$), ($\arrd$), and
($\blpi$). It corresponds to the Calculus of constructions with an
extra constant $\iota$ of type $\Box$. Adding a constant of type
$\Set$ in \lpr{}, like adding variables of type $\Set$ does not
require to introduce an extra sort $\triangle$.

Some developments in the Calculus of constructions choose to declare
the types of mathematical objects such as $\iota$, \textit{nat},
etc. in $*$, that would correspond to $\bliota:\Prop$, fully
identifying types and propositions.
The drawback of this choice is that it gives the type $*$ to
  the type $\iota$ of the constant $0$, and the type $\Box$ to the
  type $\iota \ra *$ of the constant \textit{positive}, while, in
  Simple type theory, both $\iota$ and $\iota \ra o$ are simple types.
  This is the reason why, in the theory $\cU$, we give the type $\Set$
  and not the type $\Prop$ to the constant $\bliota$.
So, the expression of
  the simple type $\iota \ra o$ uses the constant $\arrd$, that is,
  type constructors, as both $\bliota$ and $\o$ have type $\Set$, and
  not the constant $\blpi$, dependent types, that would be used if
  $\bliota$ had the type $\Prop$ and $\o$ the type $\Set$. Dependent
types, the constant $\blpi$, are thus marginalized to type functions
mapping proofs to terms.

\subsection{The Minimal sub-theory}

Adding the axioms ($\imp$) and ($\arr$) yields a sub-theory with the
12 axioms ($\Set$), ($\El$), ($\bliota$), ($\Prop$), $(\Prf)$,
($\imp$), ($\fa$), ($\o$), ($\arr$), ($\arrd$), ($\impd$), and
($\blpi$) called the ``Minimal sub-theory'' of the theory $\cU$.  It
contains both the 10 axioms of the Calculus of constructions and the 9
axioms of Minimal simple type theory.  It is a formulation of the
Calculus of constructions where dependent and non dependent arrows are
distinguished. It is not a genuine extension of the Calculus of
constructions as, each time we use a non dependent constant $\arr$ or
$\imp$, we can use the dependent ones instead: a term of the form $t
\arr u$ can always been replaced with the term $t \arrd
\tabs{x}{\El~t}{u}$, where the variable $x$ does not occur in $u$, and
similarly for the implication.  Thus, any proof expressed in the
Minimal sub-theory, in particular any proof expressed in Minimal
simple type theory, can always be translated to the Calculus of
constructions.

Conversely, a proof expressed in the Calculus of constructions can be
expressed in this theory. In a proof, every symbol $\arrd$ or $\impd$
that uses a dummy dependency can be replaced with a symbol $\arr$ or
$\imp$. Every proof that does not use $\arrd$, $\impd$ and $\blpi$,
can be expressed in Minimal simple type theory.

\subsection{The Calculus of constructions with dependent types at
  the object level}

In the Calculus of constructions with a constant $\iota$ of type
$\Box$, there are no dependent types at the object level. We have
types $\iota \ra \iota$ and $\iota \ra *$, thanks to the rule $\langle
\Box, \Box, \Box \rangle$, but no type $\iota \ra \Box$. We can introduce such
dependent types by adding an extra sort $\triangle$, together with an axiom
$\Box:\triangle$ and a rule $\langle \Box, \triangle, \triangle
\rangle$.
We obtain this way a Pure type system whose expression in \lpr{}
\cite{cousineau07tlca} contains 13 declarations and 7 rules.

Using the same notations as above, $\Prop$ for $U_{*}$,
$\Set$ for
$U_{\Box}$, $\Sett$ for
$U_{\triangle}$, etc., we get exactly
the 13 axioms
($\Prop$), $(\Prf)$,
($\Set$), ($\El$),
($\Sett$), ($\Ty$),
($\o$), ($\set$),
($\impd$),
($\blpi$),
($\fa$),
($\arrd$), and
($\dtarr$).  The theory formed with these 13 axioms is thus equivalent to
more common formulations of the Calculus of constructions with
dependent types at that object level.

\subsection{The Calculus of constructions with prenex predicative polymorphism}

In the Calculus of constructions with an extra sort $\triangle$,
polymorphism at the object level can be added with the rule
$\langle \triangle, \Box, \Box \rangle$
that allows to build terms of the form
$\tpi{x}{\Box} {x \rightarrow x}:\Box$
at the expense of making the system
inconsistent~\cite{10.1007/BFb0014058,coquand:inria00076023}.
Thus, just like in
Simple type theory, we restrict to prenex predicative
polymorphism: so, besides the sort $\triangle$, whose only closed
irreducible element is $\Box$, we introduce a sort $\diamond$ for
schemes and two rules $\langle \triangle, \Box, \diamond \rangle$ to
build schemes by quantifying over an element of $\triangle$, that is,
over $\Box$, in a type and $\langle \triangle, \diamond, \diamond
\rangle$ to build schemes by quantifying over $\Box$ in another
scheme. We also add a rule $\langle \triangle, *, * \rangle$ to
quantify over $\Box$ in a proposition.

Alternatively, the Calculus of constructions with prenex predicative
polymorphism can be defined as a cumulative type system
\cite{barras99phd}, making $\Box$ a subsort of $\diamond$ and having
just one rule $\langle \triangle, \diamond, \diamond \rangle$ to
quantify over a variable of type $\Box$ in a scheme and a rule
$\langle \triangle, *, * \rangle$ to quantify over $\Box$ in a
proposition. As there is no function whose co-domain is $\Box$, this
subtyping does not need to propagate to product types.

Expressing this Cumulative type system in \lpr{} introduces 8
declarations and 4 rules on top of the Calculus of constructions: 2
declarations for the sort $\triangle$
\begin{itemize}
\item a constant $U_{\triangle}$ of type $\TYPE$,
\item and a constant $\varepsilon_{\triangle}$ of type $U_{\triangle} \ra \TYPE$,
\end{itemize}

\noindent
1 declaration and 1 rule for the axiom $\Box:\triangle$
\begin{itemize}
\item a constant $\dot{\Box}$ of type $U_{\triangle}$,
\item and a rule
$\varepsilon_{\triangle}~\dot{\Box} \lra \Set$
  (remember that $\Set$ is $U_{\Box}$),
\end{itemize}

\noindent
2 declarations for the sort $\diamond$
\begin{itemize}
\item a constant $U_{\diamond}$, that we write $\Scheme$, of type $\TYPE$,
\item and a constant $\varepsilon_{\diamond}$, that we write $\Els$, of type
  $\Scheme \ra \TYPE$,
\end{itemize}

\noindent
1 declaration and 1 rule to express that $\Box$ is a subtype of $\diamond$
\begin{itemize}
\item a constant $\lift$ of type $\Set \ra \Scheme$,
\item and a rule $\Els~(\lift~x) \lra \El~x$ (remember that $\El$
  is $\varepsilon_{\Box}$),
\end{itemize}

\noindent
1 declaration and 1 rule for the rule $\langle \triangle, \diamond,
\diamond\rangle$
\begin{itemize}
\item a constant $\dot{\Pi}_{\langle \triangle, \diamond, \diamond\rangle}$, that we
  write $\calfa$, of type $\tpi{z}{U_{\triangle}}{(\varepsilon_{\triangle}~z \ra \Scheme) \ra \Scheme}$,
\item  and a rule
$\Els~(\calfa~x~y) \lra \tpi{z}{\varepsilon_{\triangle}~x}{\Els~(y~z)}$,
\end{itemize}

\noindent
and 1 declaration and 1 rule for the rule $\langle \triangle, *, *\rangle$
\begin{itemize}
\item $\dot{\Pi}_{\langle \triangle,*,* \rangle}$, that we write $\sffa$, of type
$\tpi{z}{U_{\triangle}}{(\varepsilon_{\triangle}~z \ra \Prop) \ra \Prop}$,
\item and a rule
$\Prf~(\sffa~x~y)
\lra \tpi{z}{\varepsilon_{\triangle}~x}{\Prf~(y~z)}$.
\end{itemize}

But, again, this theory can be simplified. Indeed, just like $\Box$ is the
only closed irreducible term of type $\triangle$, $\dot{\Box}$ is the only
closed irreducible term of type $U_{\triangle}$ and thus for any closed
term $t$ of type $U_{\triangle}$, the term
$\varepsilon_{\triangle}~t$ reduces to $\Set$.  So in the type of
the constants $\calfa$ and $\sffa$: $\tpi{z}{U_{\triangle}}{(\varepsilon_{\triangle}~z \ra \Scheme) \ra \Scheme}$ and
$\tpi{z}{U_{\triangle}}{(\varepsilon_{\triangle}~z \ra \Prop) \ra
\Prop}$, we can replace the expression $\varepsilon_{\triangle}~z$
with $\Set$.  Then, as there is no point in building a function space
whose domain is a singleton, we can simplify these type further to
$(\Set \ra \Scheme) \ra \Scheme$ and $(\Set \ra \Prop) \ra \Prop$.
Accordingly, the associated reduction rules simplify to
$$\Els~(\calfa~y) \lra \tpi{z}{\Set}{\Els~(y~z)}$$
and
$$\Prf~(\sffa~y) \lra \tpi{z}{\Set}{\Prf~(y~z)}$$

Then we can drop the symbols $\varepsilon_{\triangle}$ and $\dot{\Box}$,
the rule $\varepsilon_{\triangle}~\dot{\Box} \lra \Set$, and the
symbol $U_{\triangle}$.
We are left with the 5 declaration and 3 rules
\begin{itemize}
\item $\Scheme:\TYPE$ \hfill ($\Scheme$-decl)
\item $\Els:\Scheme \ra \TYPE$ \hfill ($\Els$-decl)
\item $\lift:\Set \ra \Scheme$ \hfill ($\lift$-decl)
\item $\Els~(\lift~x) \lra \El~x$ \hfill ($\lift$-red)
\item $\calfa:(\Set \ra \Scheme) \ra \Scheme$ \hfill ($\calfa$-decl)
\item $\Els~(\calfa~y) \lra \tpi{z}{\Set}{\Els~(y~z)}$ \hfill ($\calfa$-red)
\item $\sffa:(\Set \ra \Prop) \ra \Prop$ \hfill ($\sffa$-decl)
\item $\Prf~(\sffa~y) \lra \tpi{z}{\Set}{\Prf~(y~z)}$ \hfill
($\sffa$-red)
\end{itemize}
that is, the 5 axioms ($\Scheme$), ($\Els$), ($\lift$), ($\calfa$),
and ($\sffa$).  Adding these 5 axioms to the 10 axioms defining the
Calculus of constructions yields the 15 axioms ($\Set$), ($\El$),
($\bliota$), ($\Prop$), $(\Prf)$, ($\impd$), ($\fa$), ($\o$),
($\arrd$), ($\blpi$), ($\Scheme$), ($\Els$), ($\lift$), ($\calfa$),
and ($\sffa$) defining the Calculus of constructions with prenex
predicative polymorphism~\cite{thire20phd}.

\section{Conclusion}

The theory $\cU$ is thus a candidate for a universal theory where
proofs developed in various proof systems: HOL Light, Isabelle/HOL,
HOL 4, Coq, Matita, Lean, PVS, etc. can be expressed.  This theory can
be complemented with other axioms to handle inductive types, recursive
functions, universes,
etc. \cite{assaf15phd,thire20phd,genestier20phd}. Note however that
the various axioms currently proposed for encoding recursive functions
are based on rewriting and may be difficult to translate to systems
requiring termination proofs. Using recursors should make this much
easier.

Each proof expressed in the theory $\cU$ can use a sub-theory of the theory
$\cU$, as if the other axioms did not exist: the classical connectives
do not impact the constructive ones, propositions as objects and
functionality do not impact predicate logic, dependent types and
predicate subtyping do not impact simple types, etc.

The proofs in the theory $\cU$ can be classified according to the
axioms they use, independently of the system they have been developed
in.  Finally, some proofs using classical connectives and quantifiers,
propositions as objects, functionality, dependent types, or predicate
subtyping may be translated into smaller fragments and used in systems
different from the ones they have been developed in, making the theory
$\cU$ a tool to improve the interoperability between proof systems.

In some cases, a proof can be directly transferred from one system to
the other if it does not use some axioms. For instance,
\cite{wang16aitp} showed that many proofs coming from HOL were in fact
constructive. However, we usually need to apply some transformations
on proofs to transfer them from one sub-theory to the other. For
instance, by replacing a dependent arrow by a non-dependent one when
the second argument is not actually dependent, by applying some
morphism on type universes \cite{thire20phd}, or by trying to
eliminate some uses of the excluded middle \cite{cauderlier16lfmtp},
which is part of the axioms of Isabelle/HOL \cite{paulson21}, Lean
\cite{carneiro19} and automated theorem provers.

Some of the sub-theories of $\cU$ are known to be consistent, but one
may wonder whether the theory $\cU$ itself is consistent. We
conjecture that it is but leave this difficult problem for future
work. A solution may be to extend the model developed by the second
author in \cite{dowek17icalp} for proving the consistency of the
encoding of HOL.

\subsection*{Acknowledgments}  The authors want to thank
Michael F\"arber, C\'esar Mu\~{n}oz, Thiago Felicissimo, and Makarius
Wenzel for helpful remarks on a first version of this paper,
as well as the anonymous reviewers for their useful comments.

\bibliographystyle{alphaurl}
\bibliography{axioms}
\end{document}